\documentclass[aps,prr,citeautoscript,showpacs,showkeys,groupedaddress,superscriptaddress,twocolumn]{revtex4-2}
\RequirePackage{graphicx}
\usepackage{xcolor}
\RequirePackage{hyperref}
\usepackage{amssymb}
\usepackage{amsthm}
\usepackage{braket}
\usepackage{float}
\usepackage{enumitem}
\usepackage{placeins}
\usepackage[normalem]{ulem}

\newtheorem{theorem}{Theorem}
\newtheorem{corollary}{Corollary}[theorem]
\newtheorem{lemma}[theorem]{Lemma}

\expandafter\let\csname equation*\endcsname\relax

\expandafter\let\csname endequation*\endcsname\relax

\usepackage{amsmath}

\begin{document}

\title{Visualizing Entanglement in multi-Qubit Systems}

\author{Jonas Bley}
\affiliation{Department of Physics and Research Center OPTIMAS, RPTU Kaiserslautern-Landau, 67663 Kaiserslautern, Germany}

\author{Eva Rexigel}
\affiliation{Department of Physics and Research Center OPTIMAS, RPTU Kaiserslautern-Landau, 67663 Kaiserslautern, Germany}

\author{Alda Arias}
\affiliation{Department of Physics and Research Center OPTIMAS, RPTU Kaiserslautern-Landau, 67663 Kaiserslautern, Germany}
\affiliation{Faculty of Physics, Chair of Physics Education, Ludwig-Maximilians-Universität München (LMU Munich), 80539 Munich, Germany}

\author{Nikolas Longen}
\affiliation{Department of Computer Science and Research Initiative QC-AI, RPTU Kaiserslautern-Landau, 67663 Kaiserslautern, Germany}

\author{Lars Krupp}
\affiliation{Department of Computer Science and Research Initiative QC-AI, RPTU Kaiserslautern-Landau, 67663 Kaiserslautern, Germany}
\affiliation{Embedded Intelligence, German Research Centre for Artificial Intelligence, 67663 Kaiserslautern, Germany}

\author{Maximilian Kiefer-Emmanouilidis}
\affiliation{Department of Physics and Research Center OPTIMAS, RPTU Kaiserslautern-Landau, 67663 Kaiserslautern, Germany}
\affiliation{Department of Computer Science and Research Initiative QC-AI, RPTU Kaiserslautern-Landau, 67663 Kaiserslautern, Germany}
\affiliation{Embedded Intelligence, German Research Centre for Artificial Intelligence, 67663 Kaiserslautern, Germany}

\author{Paul Lukowicz}
\affiliation{Department of Computer Science and Research Initiative QC-AI, RPTU Kaiserslautern-Landau, 67663 Kaiserslautern, Germany}
\affiliation{Embedded Intelligence, German Research Centre for Artificial Intelligence, 67663 Kaiserslautern, Germany}

\author{Anna Donhauser}
\affiliation{Faculty of Physics, Chair of Physics Education, Ludwig-Maximilians-Universität München (LMU Munich), 80539 Munich, Germany}

\author{Stefan Küchemann}
\affiliation{Faculty of Physics, Chair of Physics Education, Ludwig-Maximilians-Universität München (LMU Munich), 80539 Munich, Germany}

\author{Jochen Kuhn}
\affiliation{Faculty of Physics, Chair of Physics Education, Ludwig-Maximilians-Universität München (LMU Munich), 80539 Munich, Germany}

\author{Artur Widera}
\affiliation{Department of Physics and Research Center OPTIMAS, RPTU Kaiserslautern-Landau, 67663 Kaiserslautern, Germany}

\date{\today}
\begin{abstract}
In the field of quantum information science and technology, the representation and visualization of quantum states and related processes are essential for both research and education. In this context, a focus lies especially on ensembles of few qubits. There exist many powerful representations for single-qubit and multi-qubit systems, such as the famous Bloch sphere and generalizations. Here, we utilize the dimensional circle notation as a representation of such ensembles, adapting the so-called circle notation of qubits and the idea of representing the $n$-particle system in an $n$-dimensional space. We show that the mathematical conditions for separability lead to symmetry conditions of the quantum state visualized, offering a new perspective on entanglement in few-qubit systems and therefore on various quantum algorithms. In this way, dimensional notations promise significant potential for conveying nontrivial quantum entanglement properties and processes in few-qubit systems to a broader audience, and could enhance understanding of these concepts as a bridge between intuitive quantum insight and formal mathematical descriptions.
\end{abstract}
\maketitle

\section{Introduction}
Genuine quantum properties are hard to visualize and hence to intuitively understand. Powerful visualizations of simple two-level, single-particle systems such as the Bloch vector representation of the density matrix have been developed to represent properties and dynamics in various situations beyond the mathematical description. Due to the extraordinary mathematical complexity of multi-qubit systems, representing many-body correlations for even two- or few-qubit systems comes along with many challenges.

Geometric representations of pure multi-qubit states and entanglement were previously addressed from the perspective of the mathematical fields of topology and geometry~\cite{Brody_2001,PhysRevA.64.062307,geometry_Quantum_States,PhysRevLett.108.230502}. Other representations include the Majorana representation depicting multi-qubit states on a Bloch sphere~\cite{Makela_2010} or, alternatively, the use of separate Bloch spheres for the non-entangled part of the system and the entangled part~\cite{PhysRevA.95.032308, PhysRevA.93.062320}. Another possibility is the generalization of the Bloch sphere to a Bloch hypersphere~\cite{PhysRevResearch.4.023120}. Also, the product operator formalism can be used to visualize the underlying processes in multi-dimensional NMR spectroscopy~\cite{SORENSEN1984163, Goldenberg2010}. Lastly, a haptic model of entanglement based on knot theory has been proposed~\cite{sym13040581}.

In all of these works, entanglement is geometrically represented. However, they are difficult to generalize to more than two- or three-qubit systems. In addition, the profound mathematical background in, e.g., topology or advanced geometry often needed to understand these models adds multiple layers of complexity. These are, however, often unnecessary in the context of quantum computing algorithms~\cite{9781107002173}. To solve the latter challenge, various two-qubit visualizations are used for educational purposes~\cite{gidney_2017,wootton_2018} and also in the context of \mbox{quantum games~\cite{meqanic_2018,ashoori_weisz_2018,hello_quantum}}. 

For more general applications, one needs to go beyond two- or three-qubit systems. Here, graphical languages like the ZX, ZW or ZH calculi, that can be seen as abstractions of circuit diagrams, are commonly used to visualize quantum states and algorithms~\cite{Backens_2014,7174913,ng2017universal,Backens_2019}. Their abstractness can be an advantage, e.g., for efficiently showing gate identities and the different possible entanglement properties of multi-qubit system~\cite{7174913}. At the same time, they require an already existing understanding of the often complex underlying concepts and processes. To acquire this understanding, explicit visualizations are necessary. One possibility is the use of generalized Wigner functions, the so-called ``DROPS representation", to represent systems of few qubits ~\cite{PhysRevA.91.042122}, usable even beyond three-qubit systems~\cite{Leiner_2020}. This operator-based choice of basis lays focus on quantum correlations and is also useful to describe time evolution of multi-qubit systems \cite{KOCZOR20191}.

Compared to this, the so-called circle notation~\cite{johnston_harrigan_gimeno-segovia_2019} makes use of the computational (0,1) basis. The aim of this visualization is to minimize the reluctance of learners towards quantum notations and linear algebra formalities, and instead highlight the basic ideas and mechanisms of quantum algorithms explicitly. In this notation, complex numbers are represented graphically by visualizing their magnitude as a filled area in a circle, and their phase as gauge in the circle. A drawback of this visualization is that the action of gate operations on the multi-qubit registers is not intuitive but rather has to be memorized. Furthermore, entanglement remains hidden. 

Another idea is to represent qubit systems in space, assigning every qubit an axis. This is akin to the idea from classical computer science to represent $n$-bit-systems in $n$-dimensional hypercubes for the development of classical error correction codes like Hamming code~\cite{bellcore_aiello, doi:10.1080/00207160211287,  doi:10.1137/S0097539798332464}. For quantum states, this is referred to in~\cite[Sec. 17]{geometry_Quantum_States} for the purpose of showing different types (W, D, GHZ,...) of entangled states as various combinations of vertices on the (hyper)-cube. So-called \textit{color codes} \cite{PhysRevLett.97.180501} utilize topological ideas and the representation of qubit systems in lattices, hypercubes or hypercube-like systems, where, sometimes, qubits are represented as colored axes in space~\cite{Kubica_2015,2017CoTPh68285A,Vasmer_2022}. Here, unitary operations are shown as operators that act along the axis of the corresponding qubit. For educative purposes, qubits and unitary operations are also represented as axes in space and along these axes in~\cite{just_2021} where the coefficients of the (computational) basis states are visualized as colored squares, constituting the so-called ``cube notation". The interactive tool~\cite{qcvis} incorporates this dimensional approach with different kinds of ways of visualizing amplitudes (``state bar plot", ``Q sphere", ``state cube" and ``phase disk state cube").

Here, we show that such explicit visualizations enable a visual criterion for entanglement. Entanglement is utilized in many quantum algorithms like the well-known quantum teleportation algorithm or quantum error correction code and it is instructive to think about entanglement properties of few-qubit systems throughout these processes. We also show that in four- and five-qubit systems, we can ``modularize" the dimensional approach, assigning only specific qubits to axes in space and by doing so, highlighting specific entanglement properties and unitary operations in complex four- and five-qubit algorithms like quantum error correction code. For these purposes, we utilize the circle notation~\cite{johnston_harrigan_gimeno-segovia_2019} and introduce dimensionality~\cite{bellcore_aiello, doi:10.1080/00207160211287,  doi:10.1137/S0097539798332464,geometry_Quantum_States,just_2021}. We call this representation \textit{dimensional circle notation} (DCN).

DCN and other such notations consider the well known theory of learning and problem solving with multiple external representations (MERs)~\cite{ainsworth2006,Hu2021} which aims to support learners' understanding by focusing not only on symbolic-mathematical or text-based representations (e.g., formulas or written text), but also on visual-graphical representations (e.g., pictures and diagrams). In addition, as we show in this work, they provide a new perspective on separability of pure multi-qubit states.
Therefore, we see its relevance as a bridge between single-particle visualization and mathematical many-body descriptions to build intuition for few-body quantum correlations.

This paper is structured as follows:
Firstly, in Sec.~\ref{sec:circ_not}, the circle notation is introduced. It is followed by the introduction of entanglement in two-qubit systems in Sec.~\ref{sec:dim_circ}. Examples in three-qubit systems using DCN are presented in Sec.~\ref{sec:three_qubits}. Then, general separability criteria derived using DCN are given in Sec.~\ref{sec:sep_crit}. We conclude in Sec.~\ref{sec:conclusions},
illustrate further extensions of DCN, like visualization of quantum algorithms in four-and five-qubit systems, introduce an interactive DCN web tool and discuss further possible applications of visualizing entanglement properties of few-particle quantum systems.

\section{Circle Notation}\label{sec:circ_not}
We start by briefly introducing the circle notation. In an $n$-qubit system, there are $2^n$ different possible basis states represented by $2^n$ circles. We will work solely in the computational basis as it is commonly used in quantum computing. Here, the basis is given by $\{\ket{i}\}$, $i\in\{0,1\}^n, \ket{i_ni_{n-1}\dots i_1}$, which defines the $n$-qubit register. Any pure $n$-qubit state $\ket{\psi}$ can be written as a superposition of these basis states:

\begin{align}
\begin{split}
    \ket{\psi} &= \alpha_{0}\ket{0\ldots 0}+\alpha_{1}\ket{0\ldots 01}\\
    &+\alpha_{2}\ket{0\ldots 010}+\ldots+\alpha_{2^n-1}\ket{1\ldots1}
\end{split}
\end{align}

with $\alpha_i \in \mathbb{C}, \sum_{i=0}^{2^n-1}|\alpha_i|^2=1$.
As per the convention used here, the rightmost entry in the ket state corresponds to the first qubit and the leftmost entry to the $n$'th qubit. This means that the least significant qubit in the binary system corresponds to the rightmost entry. As shown in Fig.~\ref{fig:circnot}, the circle notation graphically represents the magnitudes of the amplitudes $\alpha_i$ as filled inner circles with radius $|\alpha_i|$ and their phase $\varphi$ of $\alpha_i=e^{i\varphi}|\alpha_i|$ as the angle between the radial line and a vertical line. Some important single qubit operations (in a single qubit system) are shown in Fig.~\ref{fig:single_qubit_operations} in Appendix~\ref{sec:CNsinglequbit}.

\begin{figure}[htb]
    \centering
    \includegraphics[width=0.25\textwidth]{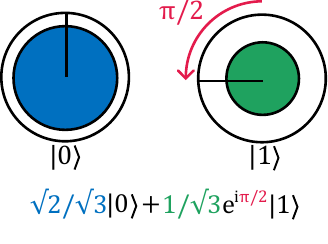}\caption{A qubit in the state $\ket{\psi}=\sqrt{2/3}\ket{0}+1/\sqrt{3}e^{i\pi/2}\ket{1}$ in circle notation~\cite{johnston_harrigan_gimeno-segovia_2019}. The outer circles represent the basis states $\ket{0}$ and $\ket{1}$. The radii of the inner circles represent the absolute value of the corresponding coefficients. The radius of the blue circle is $\sqrt{2/3}$ and the radius of the green circle $1/\sqrt{3}$. The blue area is double the size of the green area, showing that measuring would, on average, yield the result 0 twice as often as 1. The angles of the lines in respect to a vertical line represent the phases of the corresponding coefficients. Here, the angle of the line of the coefficient $1/\sqrt{3}e^{i\pi/2}$ of the basis state $\ket{1}$ is horizontal and facing left, representing the phase $\pi/2$.}
    \label{fig:circnot}
\end{figure}

For two qubits, the possible states are lined up as shown in Fig.~\ref{fig:multi_circnot}. In standard circle notation, one can not immediately determine whether the represented state is separable or entangled. We refer to~\cite{johnston_harrigan_gimeno-segovia_2019} for a precise and comprehensive introduction to the circle notation, in particular, unitary operations and measurements in multi-qubit systems. For calculating their effect, if not memorized, operations require the additional effort of checking each basis state in Dirac ket notation which could reduce the advantage of this representation in respect to the mathematical representation. 

\begin{figure}[htb]
    \centering
    \includegraphics[width=0.4\textwidth]{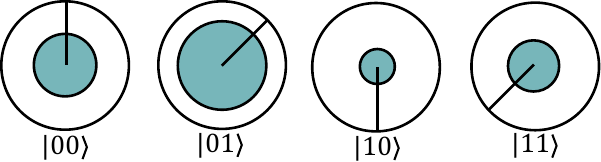}
    \caption{The two-qubit state $\ket{\psi}=1/2\ket{00}+1/\sqrt{2}e^{-i\pi/4}\ket{01}-1/\sqrt{12}\ket{10}+1/\sqrt{6}e^{3i\pi/4}\ket{11}$ in circle notation~\cite{johnston_harrigan_gimeno-segovia_2019}. The states are ordered in ascending order in the binary system, where e.g. the first qubit represents the rightmost number $i_1$ in a state $\ket{i_2i_1}$.}
    \label{fig:multi_circnot}
\end{figure}

\section{Entanglement in two-Qubit Systems}\label{sec:dim_circ}

In dimensional notations, instead of arranging states in a row, every qubit is assigned to an axis in a new direction in space. Here, in contrast to the standard circle notation, it is enough to understand these operations in single-qubit systems to understand them in any multi-qubit system~\cite{just_2021}. Additionally, as we show in this section, such visualizations reveal entanglement properties in two-qubit systems. For this, we use DCN based on circle notation~\cite{johnston_harrigan_gimeno-segovia_2019}. Fig.~\ref{fig:dimcirc2} shows how product states are formed in such a dimensional arrangement, following the standard Kronecker product.

\begin{figure*}[htb]
    \centering
    \includegraphics[width=0.85\textwidth]{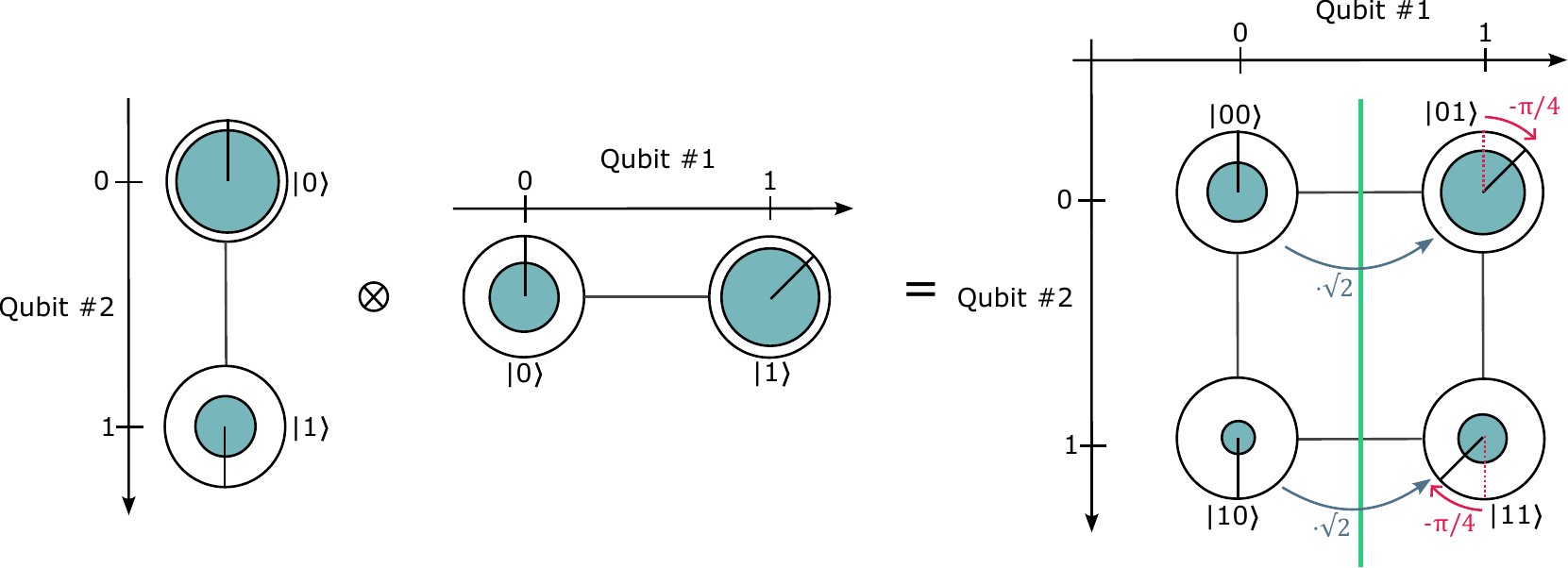}
    \caption{Visual representation of the separability of the product state $\ket{\psi}=(\sqrt{3}/2\ket{0}-1/2\ket{1})\otimes(1/\sqrt{3}\ket{0}+\sqrt{2}/\sqrt{3}e^{-i\pi/4}\ket{1})=1/2\ket{00}+1/\sqrt{2}e^{-i\pi/4}\ket{01}-1/\sqrt{12}\ket{10}+1/\sqrt{6}e^{3i\pi/4}\ket{11}$ in DCN. Similarly to \cite{just_2021}, Qubit~\#1 is attached to each of the basis states of qubit~\#2 to form a two dimensional array of basis states. The amplitudes of the combined state follow the standard Kronecker product. The system is separable which is shown by the green symmetry axis. The radii of the inner circles are compared in blue gray and the phases are compared in red. The coefficient ratio along this axis is $\alpha_{01}/\alpha_{00}=\alpha_{11}/\alpha_{10}=\sqrt{2}e^{-i\pi/4}$. Along the other symmetry axis, $\alpha_{10}/\alpha_{00}=\alpha_{11}/\alpha_{01}=1/3e^{i\pi}=-1/3$.
    }
    \label{fig:dimcirc2}
\end{figure*}

\paragraph{Entanglement} In the classical circle notation, see Fig.~\ref{fig:multi_circnot}, it is cumbersome to distinguish a separable state from an entangled one. In this section, we will show how dimensional notations allow spotting separable states in the two-qubit case. A state $\ket{\psi}=\alpha_{00}\ket{00}+\alpha_{01}\ket{01}+\alpha_{10}\ket{10}+\alpha_{11}\ket{11}$ is separable into $\ket{\psi}=(\alpha_1\ket{0}+\beta_1\ket{1})\otimes \left(\alpha_2\ket{0}+\beta_2\ket{1}\right)$, where $\otimes$ is the Kronecker product, if and only if 
\begin{equation}\label{eq:twoqubits_sep}
    \alpha_{00}\alpha_{11}=\alpha_{01}\alpha_{10}
\end{equation}
as stated in, e.g., Ref.~\cite[p. 396]{geometry_Quantum_States}. We can represent this condition in terms of coefficient ratios $\alpha_{00}/\alpha_{01}=\alpha_{10}/\alpha_{11}$ in the case of $\alpha_{01},\alpha_{11}\neq0$ or $\alpha_{10}/\alpha_{00}=\alpha_{11}/\alpha_{01}$ in the case of $\alpha_{10},\alpha_{11}\neq0$. In the case of more than two coefficients being 0, the system is trivially separable. In summary, this means we can visually not only identify entangled states, but also get a sense for the degree and the type (phase or magnitude) of entanglement by comparing the ratios of the coefficients $\alpha_{00}/\alpha_{01}=r_1e^{i\varphi_1}, \alpha_{10}/\alpha_{11}=r_2e^{i\varphi_2}$ in terms of the ratio of their amplitudes $r_1/r_2$ and the difference of their phases $\varphi_1-\varphi_2$. For example, the concurrence $\mathcal{C}$ is a common way to measure entanglement~\cite{wootters2001entanglement}. It is defined as $\mathcal{C}=2|\alpha_{11}\alpha_{00}-\alpha_{10}\alpha_{01}|=2r_1|1-r_2/r_1 e^{i\varphi_1-\varphi_2}|$ for pure two-qubit states (under the assumption of $r_1>0$). It can be seen that the concurrence is large for large differences in phases ($|\varphi_2-\varphi_1|\approx\pi$) and large or small ratios of magnitudes ($r_2/r_1\gg1$ or $r_2/r_1\ll1$). We compare these ratios for every pair of states along the axis of one qubit, where both of the corresponding coefficients are non-zero. Then, we can determine whether the system is symmetrical along that axis, apart from a (complex) ratio. If we find symmetry, we know that the system is separable. Fig.\ref{fig:dimcirc2} visualizes that building product states results in a separable system using this ratio characterization. Examples of amplitude- and phase-entangled systems are shown in Fig.~\ref{fig:dimcircsymm2}.

It is important to note that this representation of separability into single-particle states only holds if the chosen basis states are themselves separable. We consider exclusively the computational basis here, but in principle any \textit{separable} basis can be used.

\begin{figure*}[htb]
    \centering
    \includegraphics[width=0.65\textwidth]{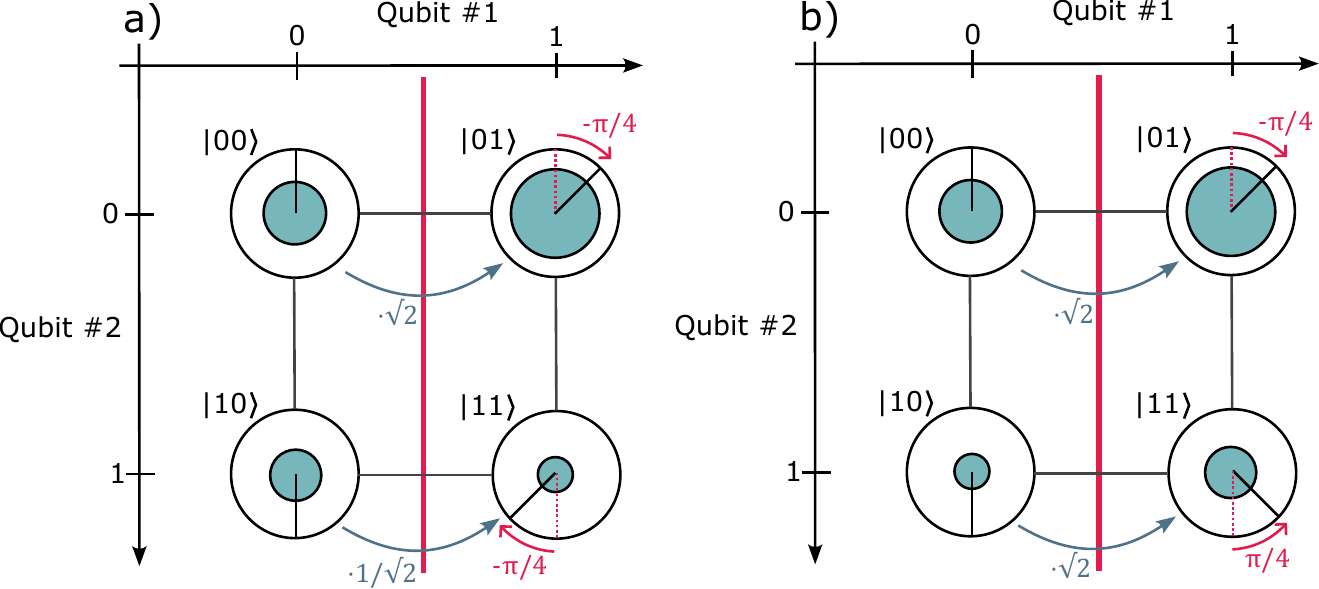}
    \caption{Entanglement in two-qubit systems visualized. a) The state $\ket{\psi}=1/2\ket{00}+1/\sqrt{2}e^{-i\pi/4}\ket{01}-1/\sqrt{6}\ket{10}+1/\sqrt{12}e^{3i\pi/4}\ket{11}$. It is entangled, because $\alpha_{11}/\alpha_{10}=1/\sqrt{2}e^{-i\pi/4}\neq\sqrt{2}e^{-i\pi/4}=\alpha_{01}/\alpha_{00}$. b) The state $\ket{\psi}=1/2\ket{00}+1/\sqrt{2}e^{-i\pi/4}\ket{01}-1/\sqrt{12}\ket{10}+1/\sqrt{6}e^{-3i\pi/4}\ket{11}$. It is (phase-)entangled, because $\alpha_{11}/\alpha_{10}=\sqrt{2}e^{i\pi/4}\neq\sqrt{2}e^{-i\pi/4}=\alpha_{01}/\alpha_{00}$.}
    \label{fig:dimcircsymm2}
\end{figure*}

\paragraph{Measurements}
Measuring a single qubit, the state collapses into a classical bit of 0 or 1. Similarly, in a terminal measurement of $n$-qubits the system collapses into the classical bit string $i=i_ni_{n-1}\dots i_1$, where $i\in \{0,1\}^n$.
The measurement of a subset of qubits is, however, more peculiar. In conventional circle notation, see Fig.~\ref{fig:multi_circnot}, one needs to precisely identify the subset of qubits measured, by evaluation of the corresponding register state, see~\cite{johnston_harrigan_gimeno-segovia_2019} for more details. In~\cite{just_2021}, measurements are shown as removal of axes. We show partial measurement (see Fig.~\ref{fig:measurement}) such that all circles along the measured qubit differing from the measured value turn empty. Afterwards, the state simply has to be renormalized.
Furthermore, the probabilities of measuring 0 or 1 are given by the sum of the areas of the inner circles of the basis states corresponding to that value. This dimensional visualization of quantum states shows an important property of entangled states: If the state is entangled, the state after the measurement will differ depending on the measurement result.

\begin{figure*}[htb]
    \centering
    \includegraphics[width=0.8\textwidth]{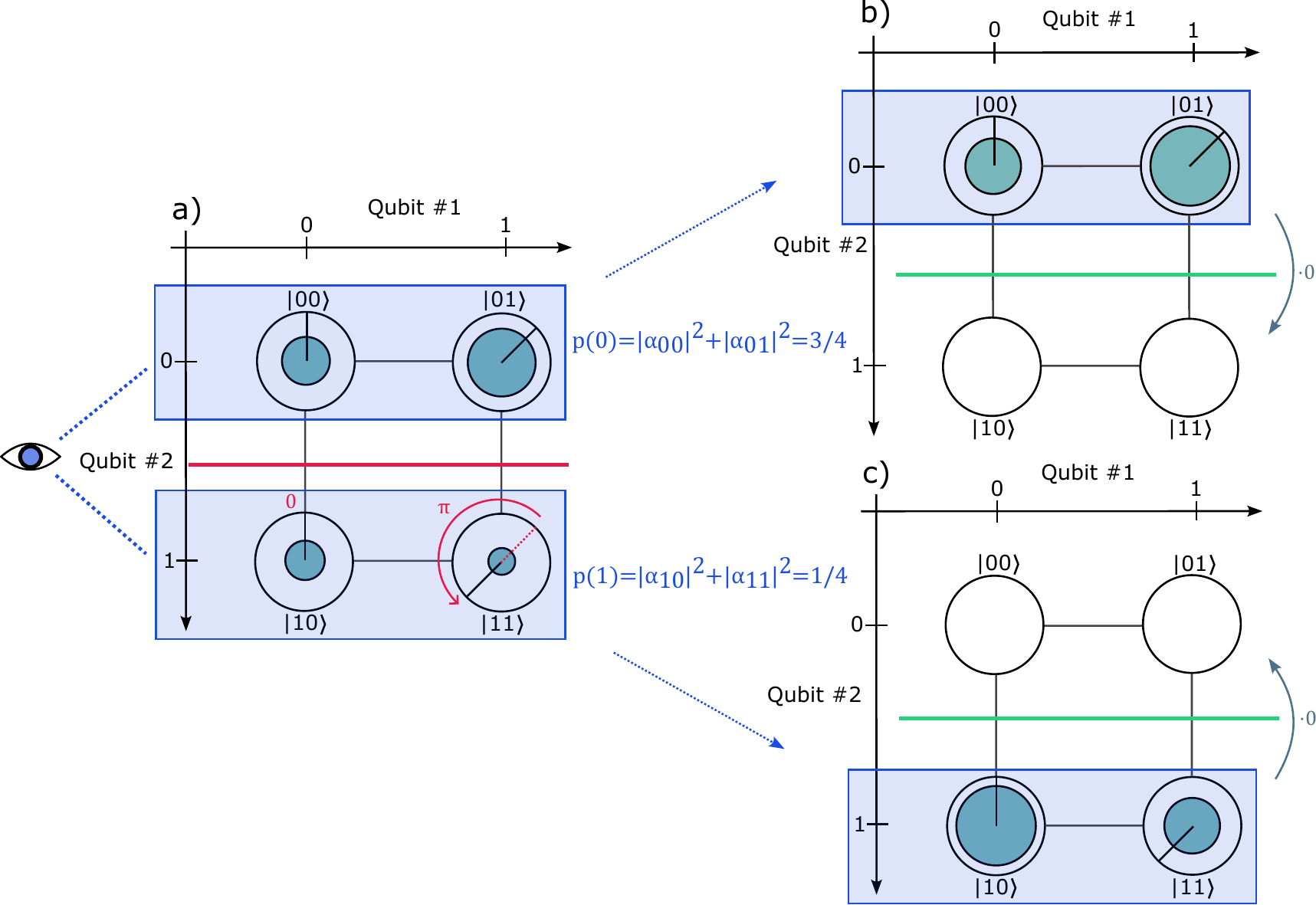}
    \caption{Entanglement in a two-qubit quantum system in the context of measurement. a) The initial state $\ket{\psi}=1/2\ket{00}+1/\sqrt{2}e^{-i\pi/4}\ket{01}+1/\sqrt{6}\ket{10}+1/\sqrt{12}e^{3i\pi/4}\ket{11}$ of the system. The state is entangled as can be seen e.g. when comparing the phase differences from top to bottom along the red axis. The phases differ by a factor $\pi$, meaning that the possible resulting states of the measurement will differ by a relative phase $\pi$. Comparing the areas of the inner circles, one can see that measuring $0$ is more likely than measuring $1$. In fact, $p(0)=(1/2)^2+(1/\sqrt{2})^2=3/4$ and $p(1)=1/4$. b) The state of the system after measuring 0. All circles where qubit~\#2 is 1 are cleared and the system is renormalized. c) The state of the system after measuring 1. After measurement, the entanglement is destroyed as can be seen in both b) and c) with the ratio 0 along the green axis.}
    \label{fig:measurement}
\end{figure*}

\paragraph{Unitary Operations}
Examples of unitary operations in single qubit systems are shown in Fig.~\ref{fig:single_qubit_operations} in Appendix~\ref{sec:CNsinglequbit} as in~\cite{johnston_harrigan_gimeno-segovia_2019}. Having understood them and in order to generalize from single-qubit systems to multi-qubit systems in circle notation, one still needs to memorize not only the effects of single qubit operations but instead all possible actions of single qubit gates on all possible qubits. The dimensional arrangement eliminates this drawback. Single-qubit gates need only to be applied alongside the axis of the qubit considered. Thus, the visualization of single-qubit operations within two-qubit systems is transferable from the one-qubit case. This, importantly, still holds for larger qubit systems as we show in the following sections. A comparison of DCN with the standard circle notation is shown in Fig.~\ref{fig:xgatemulti} for the Pauli-$X_1$- and $X_2$-gates. Note that local unitary operations leave the ratio characterization of separable states intact, i.e., we can not entangle a non-entangled system locally and vice-versa, in agreement with the no-communication theorem~\cite{Peres_2004}.

\begin{figure*}[htb]
    \centering
    \includegraphics[width=0.7\textwidth]{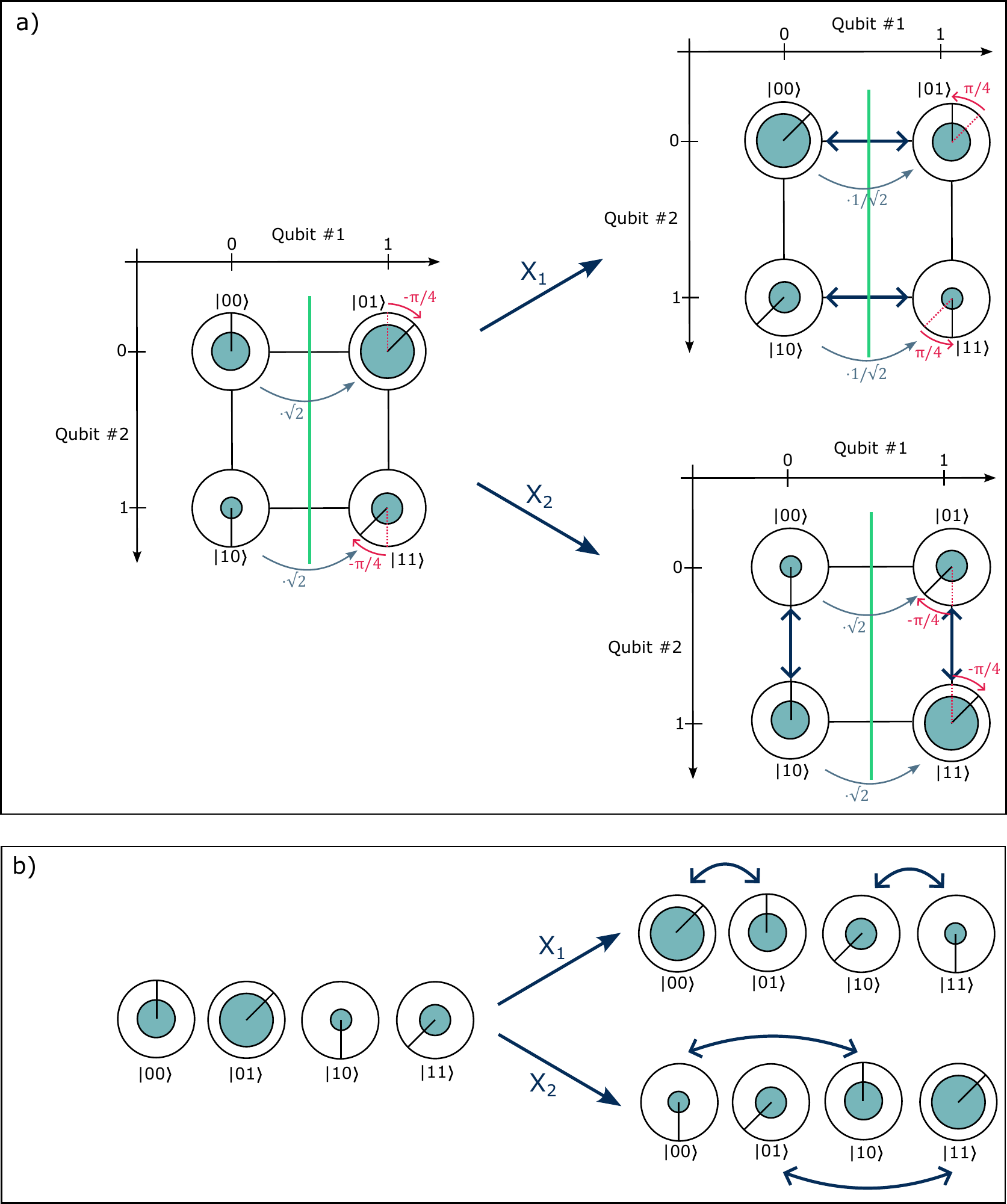}\caption{a) Visualization of an example of the conservation of entanglement properties under single-qubit unitary operations. $X_1$- and $X_2$-gate acting dimensionally on the two-qubit state $\ket{\psi}=1/2\ket{00}+1/\sqrt{2}e^{-i\pi/4}\ket{01}-1/\sqrt{12}\ket{10}+1/\sqrt{6}e^{3i\pi/4}\ket{11}$, similarly to \cite{just_2021}. The $X_1$-gate acts on \textit{all} states along the axis of qubit~\#1, swapping the coefficients of $\ket{00}$ and $\ket{01}$ as well as the coefficients of $\ket{10}$ and $\ket{11}$. The outcome of the $X_1$ operation is the state $\ket{\psi}=1/\sqrt{2}e^{-i\pi/4}\ket{00}+1/2\ket{01}+1/\sqrt{6}e^{3i\pi/4}\ket{10}-1/\sqrt{12}\ket{11}$. Similarly, the $X_2$ gate acts on all states along the axis of qubit~\#2, swapping the coefficients of $\ket{00}$ and $\ket{10}$ as well as $\ket{01}$ and $\ket{11}$. The outcome of the $X_2$ operation is the state $\ket{\psi}=-1/\sqrt{12}\ket{00}+1/\sqrt{6}e^{3i\pi/4}\ket{01}+1/2\ket{10}+1/\sqrt{2}e^{-i\pi/4}\ket{11}$. In both cases, due to the unitary operations being local, the separability of the system is retained as is shown by the green symmetry axes. b) The same operations in standard circle notation for comparison~\cite{johnston_harrigan_gimeno-segovia_2019}. Here, separability and its retention is not as easily visible.}
    \label{fig:xgatemulti}
\end{figure*}

Two-qubit operations also work geometrically in dimensional notations and -- again -- avoid the necessity of memorizing multiple operations of, e.g., controlled gates where the targeted and controlled qubits are swapped.

The CNOT-gate applies a NOT (X)-gate to the target qubit if the control qubit has value 1. This has a geometric explanation: the CNOT-gate swaps all states where the control qubit is 1 along the axis of the target qubit. The CNOT-gate is a multi-qubit gate that can change entanglement properties of the system.

The SWAP-gate exchanges two qubits in the system, which is equivalent to swapping the two qubit axes, while entanglement properties of the system are conserved. This gate can be decomposed into three CNOT-gates which is relevant in practice e.g. because existing quantum computer hardware can often only make use of CNOT-gates for qubit interactions. Fig.~\ref{fig:cnotswap} shows how this decomposition can be visualized geometrically. The conservation of entanglement properties is apparent in dimensional notations, because swapping axes does not change which ratio of coefficients are present in the system.

\begin{figure*}[htb]
    \centering
    \includegraphics[width=0.8\textwidth]{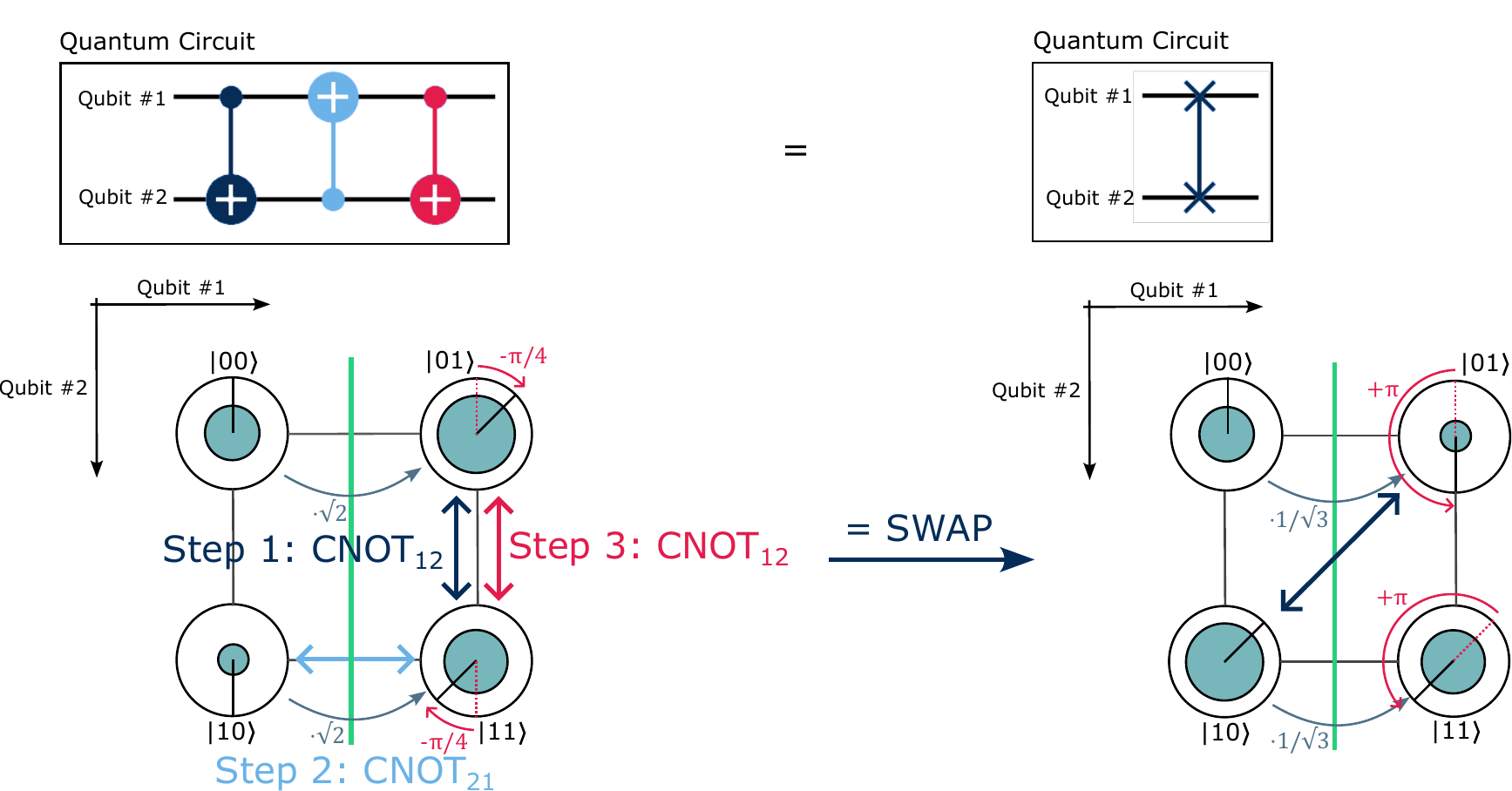}
    \caption{The SWAP-gate as an example of a multi-qubit gate that retains entanglement properties, implemented as three CNOT-gates on the two-qubit state $\ket{\psi}=1/2\ket{00}+1/\sqrt{2}e^{-i\pi/4}\ket{01}-1/\sqrt{12}\ket{10}+1/\sqrt{6}e^{3i\pi/4}\ket{11}$. The corresponding Quantum Circuits are displayed at the top. The relation CNOT$_{12}$CNOT$_{21}$CNOT$_{12}=$SWAP can be geometrically explained by swapping the states along the two axes step by step. The final state $\ket{\psi}=1/2\ket{00}-1/\sqrt{12}e^{-i\pi/4}\ket{01}+1/\sqrt{2}e^{-i\pi/4}\ket{10}+1/\sqrt{6}e^{3i\pi/4}\ket{11}$ is shown on the right hand side. Note that it is visually apparent that both states are separable (see Fig.~\ref{fig:dimcircsymm2} a)), which the SWAP-gate does not change, but a single CNOT-gate would entangle the system.}
    \label{fig:cnotswap}
\end{figure*}

In Appendix~\ref{sec:DCNmulti}, we provide additional DCN examples for $\text{CNOT}_{12}=(\text{H}_2\otimes \text{H}_1)\text{CNOT}_{21}(\text{H}_2\otimes \text{H}_1)$ as an example of a phase kickback swapping the role of target and control qubit, see Fig.~\ref{fig:phasekickback}. We also show the Deutsch algorithm which is often considered as an example of quantum parallelism and a (albeit non-practical) use-case of phase kickback, see Fig.~\ref{fig:deutsch}. The representation of the Deutsch algorithm in DCN shows that although a CNOT-gate is present, no entanglement has been created, and therefore the algorithm could, in principle, be realized classically which has been shown in classical optical systems~\cite{Kish2023}. Again, we expect dimensional approaches to be more intuitive than non-dimensional approaches.

\section{Entanglement in three-Qubit Systems}\label{sec:three_qubits}
We now shift from two-qubit systems to three-qubit systems and explore the advances of dimensional notations in respect to standard circle notation. Similarly to the transfer from one-qubit systems to two-qubit systems, dimensional operations in three-qubit systems are transferable from the one- or two-qubit cases. Still, the additional qubit leads to a few key differences that we will explain in the following.

\paragraph{(Partial) Separability and Entanglement}
To distinguish separable states from entangled states, we apply a similar procedure taken from the two-qubit case to determine whether a three-qubit system is separable. The two key differences are:

\begin{enumerate}
    \item In order to compare the ratios of coefficients, we look for symmetry \textit{planes} instead of axes. This way, we compare the ratios of the top coefficients with the bottom coefficients, left with right or front with back (see Corollary~\ref{cor:par_sep}, Appendix~\ref{sec:sep}). This is shown in Fig.~\ref{fig:dimcircsymm3}.
    \item We can differentiate between partial and full separability and compare along two planes. If the ratios are the same along only one plane, we have an entangled two-qubit system that the third qubit, represented by the axis perpendicular to this symmetry plane, is independent of (Fig.~\ref{fig:dimcircsymm3} is an example of such a state). If and only if the ratios are the same along two planes, they are also the same along the third plane and we have a fully separable system. % (see Appendix~\ref{sec:sep}, Corollary~\ref{cor:ratios})
\end{enumerate}

\begin{figure}[htb]
    \centering
    \includegraphics[width=0.9\columnwidth]{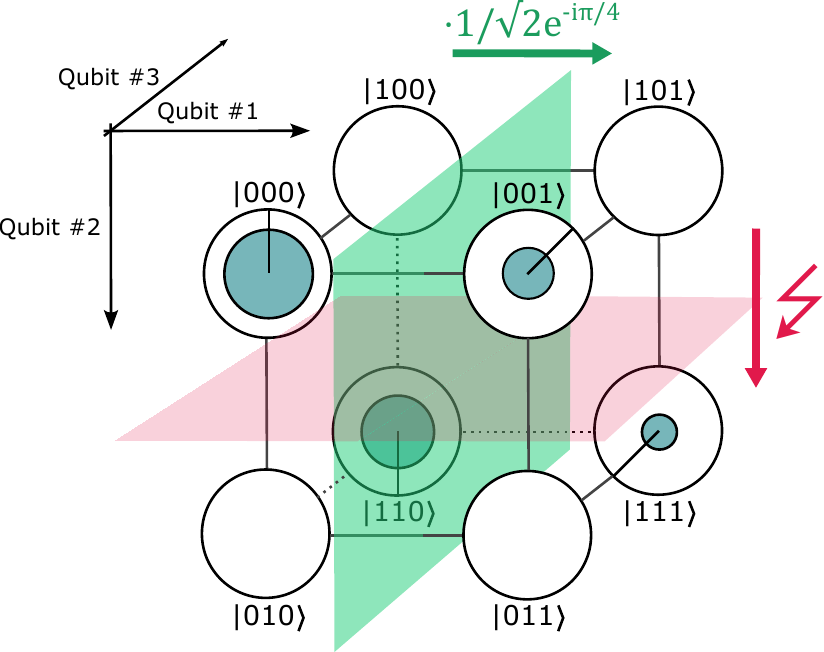}
        \caption{The partially separable state $\ket{\psi}=1/\sqrt{2}\ket{000}+1/\sqrt{6}e^{-i\pi/4}\ket{001}-1/2\ket{110}+1/2\sqrt{3}e^{3i\pi/4}\ket{111}=(\sqrt{2}/\sqrt{3}\ket{00}-1/\sqrt{3}\ket{11})\otimes (\sqrt{3}/2\ket{0}+1/2e^{-i\pi/4}\ket{1})$.
        The green symmetry plane shows that qubit~\#1 can be separated from the system as the ratio $1/\sqrt{2} e^{-i\pi/4}$ can be applied to go from left to right. However, the system is \textit{not} fully separable because there is no symmetry along the axis of qubit~\#2, i.e. there is no such ratio and the red plane is not a symmetry plane. Similarly, there is no symmetry along the axis of qubit~\#3.
    }
    \label{fig:dimcircsymm3}
\end{figure} 

\paragraph{Quantum Teleportation}
Quantum Teleportation has been at the heart of quantum technologies for many years, allowing the transfer of quantum information between two parties over arbitrary distances when an EPR pair is shared between them. It has multiple applications in quantum communication~\cite{Gisin2007} and quantum computation~\cite{10.1145/3394885.3431604,PhysRevLett.124.080502} and is therefore an essential part of quantum information processing~\cite{Pirandola2015}. Because it incorporates many fundamental concepts of quantum information science and technology like entanglement, unitary operations and measurements, quantum teleportation is a suitable example of how dimensional notations could enhance understanding of quantum algorithms in general \cite{just_2021}. In particular, the protocol utilizes three-qubit entanglement which, as we show, can be visualized in dimensional notations, offering a visual perspective on entanglement within the protocol.

Quantum teleportation works as follows: A pair of entangled qubits \#2 and \#3 in the state $\ket{\phi^+}_{32}=1/\sqrt{2}(\ket{00}+\ket{11})$ is prepared. Qubit~\#3 is sent to Bob and qubit~\#2 to Alice. Alice also has qubit~\#1 in the state $\ket{\psi_1}$ which she does not necessarily need to know and that she wants to teleport to Bob.

During quantum teleportation, the information of qubit~\#1 is transferred to qubit~\#3. In dimensional notations, this has geometric meaning (see Fig.~\ref{fig:qt_sum}): Because of the equivalence of an axis with a qubit, transferring information from one qubit to another is the same as transferring information from one \textit{axis} to another. This can be done using the unitary operations CNOT$_{12}$ and $H_1$. These operations only act on qubit~\#1 and \#2, i.e. along the axis of qubit~\#1 and \#2. In practice, this means that Alice does not need physical access to qubit~\#3. Note that this transfer of information is only possible because qubit~\#2 and qubit~\#3 are entangled. To achieve her goal, Alice first applies a CNOT-gate with qubit~\#1 as control and qubit~\#2 as target, fully entangling the system (see Fig.~\ref{fig:qt_sum}). She then applies a Hadamard-gate to qubit~\#1.

\begin{figure*}[htb]
    \centering
    \includegraphics[width=0.8\textwidth]{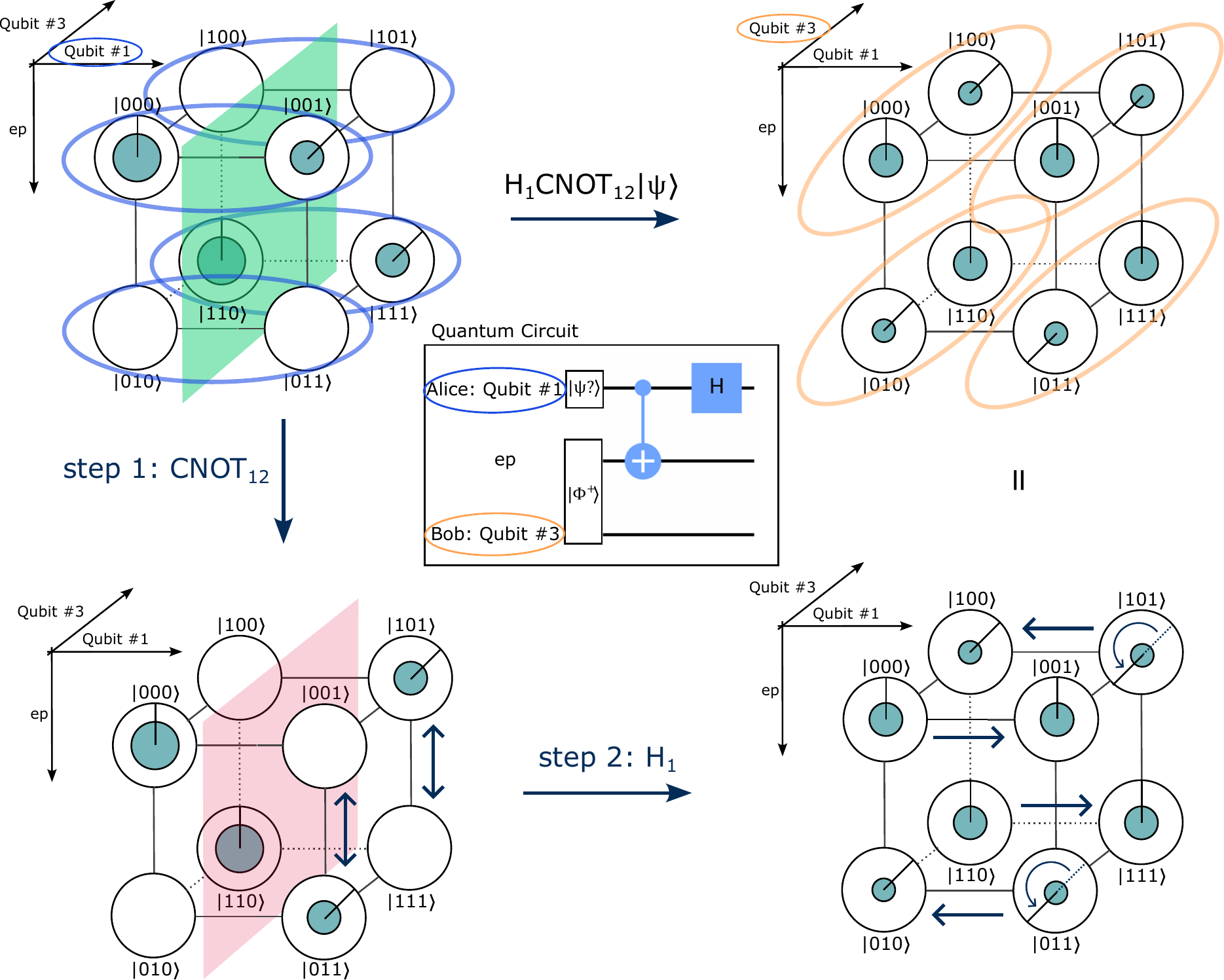}
    \caption{Entanglement utilization in the central part of the quantum teleportation algorithm. Qubit~\#1 starts in the arbitrary state $\ket{\psi}_1=\sqrt{2/3}\ket{0}+1/\sqrt{3}e^{-i\pi/4}\ket{1}$. Qubits~\#2 and \#3 start in the bell state $\ket{\psi^+}_{32}=1/\sqrt{2}(\ket{00}+\ket{11})$. The product state is partially separable (the same as the state in Fig.~\ref{fig:dimcircsymm3}), shown by the green symmetry plane. The information that is initially stored in qubit~\#1 (blue) which is independent of the other two qubits is transferred to qubit~\#3 (yellow) using only unitary operations on qubit~\#1 and \#2, i.e. operations only along the axes of qubits \#1 and \#2 in two steps. Step 1: Swap states on the right hand side (where qubit~\#1 is 1) along the axis of qubit~\#2 using a CNOT-gate with control qubit~\#1 and target qubit~\#2. This destroys the symmetry along the axis of qubit \#1 and the system is now fully entangled, hence the red plane is shown. In fact, one can see that entangling qubit~\#1 with qubit~\#2 also entangled qubit~\#3 with qubit~\#1 by transfer of entanglement. Step 2: Split states along axis of qubit~\#1 using a Hadamard-gate on qubit~\#1. The application of the Hadamard-gate does not change the entanglement properties of the system as it is a local unitary operation. The corresponding quantum circuit is displayed in the middle.}
    \label{fig:qt_sum}
\end{figure*}

When Alice now measures qubit~\#1 and qubit~\#2, the four possible measurement outcomes 00, 01, 10 and 11 lie on the 2D plane spanned by qubit~\#1 and qubit~\#2. The resulting state of qubit~\#3 depends on the measurement result. Alice sends the result to Bob who applies an $X$ and/or a $Z$-gate if needed so that his qubit~\#3 is in the state that qubit~\#1 previously was in. This last step is shown in Fig.~\ref{fig:qt_measurement}, Appendix~\ref{sec:DCNmulti}.

\section{Entanglement in multi-Qubit Systems}\label{sec:sep_crit}

We now generalize the visualization of entanglement properties to multi-qubit systems beyond three qubits. While determining whether a given mixed state is separable is, in general, a (strongly) NP-hard problem~\cite{gurvits2003classical,gharibian2009strong}, separability of pure states is well understood within the density matrix formalism \cite{Plenio_1998,Horodecki_2009}.
However, as quantum computing algorithms tend to be discussed quite explicitly~\cite[Sec. 4]{9781107002173}, it can be tedious to refer back to the density matrix formalism while working out a given quantum algorithm, just to show the entanglement properties of the system at any given moment. We argue that it would be beneficial to be able to see entanglement even when working without the density matrix formalism. In the ideal case, the system containing all qubits that are part of the algorithm is a pure state. Lemma~\ref{lem} gives a necessary criterion for separability of a given pure state in the computational basis~\cite{doi:10.1142/S0129054103002035, dyn_ent}.

\begin{lemma}[Necessary Criterion for Separability]\label{lem}
    Let $\ket{\psi_N}=\sum_{i=0}^{N-1}c_i\ket{i}$ $\in\mathcal{H}_N$ be a $N=PQ$-separable pure state and $r_0\in[0,Q-1]$, $k_0\in[0,P-1]$. Then, for all $r\in[0,Q-1]$ and $k\in[0,P-1]$: $c_{k_0Q+r_0}c_{kQ+r}=c_{k_0Q+r}c_{kQ+r_0}$.
\end{lemma}

We provide an alternative proof in Appendix~\ref{sec:sep}. In fact, this criterion can be adjusted to be sufficient by specifying the first non-zero coefficient $c_{i_0}$ with $i_0=k_0Q+r_0$ to examine PQ-separability ~\cite{doi:10.1142/S0129054103002035}. The criterion given in~\cite{doi:10.1142/S0129054103002035} only considers states where $c_{kQ+r}=0$ for all $r<r_0$ (see Corollary~\ref{cor:extra0}, Appendix~\ref{sec:sep}). In fact, in DCN, one can see that states where $c_{kQ+r}\neq0$ for some $r<r_0$ (even if $k>k_0$), are not separable, because symmetry is absent in these states. This is shown in Fig.~\ref{fig:them_sep}, where 2-4-separability is examined for the state $\ket{\psi}=\frac{1}{\sqrt{6}}(\ket{1}+\ket{2}-\ket{3}+\ket{4})+\frac{1}{2\sqrt{3}}(e^{-i\pi/2}\ket{5}+e^{i\pi/2}\ket{7})$ in DCN. As can be seen in DCN, this state is not 2-4 separable, even though $c_1c_7=c_3c_5$, because of the absence of symmetry in regard to qubit~\#3. Therefore, it can be seen in dimensional notations that the generalization of the criterion given in~\cite{doi:10.1142/S0129054103002035} is the following:

\begin{theorem}[PQ-Separability of pure States]\label{them:main}
    Let $\ket{\psi_N}=\sum_{i=0}^{N-1}c_i\ket{i}$ be a pure state where for some $i_0=k_0Q+r_0\in[0,N-1]$, $\alpha_{i_0}\neq0$ and $\forall i < i_0, \alpha_i=0$. Then, $\ket{\psi}$ is $N=PQ$-separable if and only if for all $k\in[k_0+1,P-1]$ and $r\in[0,Q-1]$: $c_{k_0Q+r_0}c_{kQ+r}=c_{k_0Q+r}c_{kQ+r_0}$.
\end{theorem}

We provide a proof in Appendix~\ref{sec:sep}. Theorem~\ref{them:main} has alternative formulations that are also given in Appendix~\ref{sec:sep}. While the forward implication part of the proof is given by Lemma~\ref{lem}, the idea of the backwards implication is finding ratios $m_r$ such that $c_{kQ+r}=m_r c_{kQ+r_0}$ for all $k\in[k_0+1,P-1]$. We therefore have Corollary~\ref{cor:ratios} \cite{dyn_ent} as part of the proof of Theorem~\ref{them:main}.

\begin{corollary}[Ratio Characterization of Separability]\label{cor:ratios}
    Let $\ket{\psi_N}=\sum_{i=0}^{N-1}c_i\ket{i}$ be a pure state where for some $i_0=k_0Q+r_0\in[0,N-1]$, $c_{i_0}\neq0$ and $\forall i < i_0, c_i=0$. Then, $\ket{\psi}$ is PQ-separable if and only if there exist ratios $m_r\in\mathbb{C}$ for $r\in[r_0+1,Q-1]$ such that for all $k\in[k_0+1,P-1]$, $\ket{\psi}$ can be written in the form $c_{kQ+r}=m_r c_{kQ+r}$ for $r\in[r_0+1,Q-1]$ and $c_{kQ+r_0}=0$ for $r<r_0$.
\end{corollary}

These are exactly the ratios that are visible in dimensional notations. The visual implications of this characterization of separability can be seen throughout this paper. Fig.~\ref{fig:them_sep} shows in DCN why $c_{kQ+r_0}=0$ for $r<r_0$ for all $k\in[k_0+1,P-1]$ is a necessary condition for separability. In Appendix~\ref{sec:4-4-sep}, Fig.~\ref{fig:4-4-sep}, we show 4-4 separability in a four-qubit system using this criterion.

\begin{figure}[htb]
    \centering
    \includegraphics[width=0.48\textwidth]{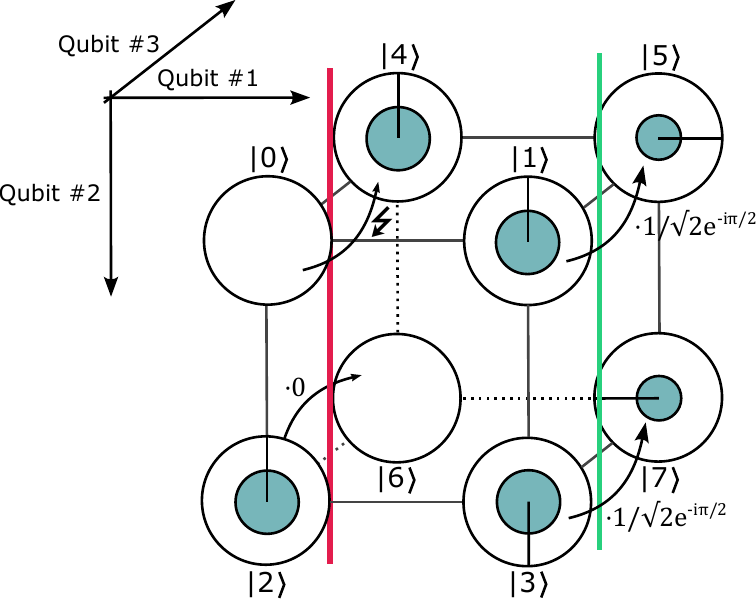}
    \caption{Examination of 2-4-separability (separability in regard to qubit~\#3) of the state $\ket{\psi}=\frac{1}{\sqrt{6}}(\ket{1}+\ket{2}-\ket{3}+\ket{4})+\frac{1}{2\sqrt{3}}(e^{-i\pi/2}\ket{5}+e^{i\pi/2}\ket{7})$ in DCN. It is visible that the state is not separable (in particular not $PQ=$ 2-4-separable) as the red axis is not a symmetry axis, showing that $c_{kQ+r}=0$ for all $r<r_0$ and all $k$ is a necessary condition in Corollary~\ref{cor:ratios} (here, $c_{4}=c_{4+0}\neq0$) in the case of $c_0=0$. Applying Corollary~\ref{cor:ratios} to examine $PQ=$ 2-4-separability, we find $i_0=r_0=1$, $k_0=0$. If we only checked for $c_{k_0Q+r_0}c_{kQ+r}=c_{k_0Q+r}c_{kQ+r_0}$ for $r=k=1$, i.e. $c_{1}c_{7}=c_{3}c_{5}$, we would find that the state is separable.}
    \label{fig:them_sep}
\end{figure}

\section{Conclusions \& Further Extensions}\label{sec:conclusions}

The standard circle notation is already a useful tool for introductory quantum computing courses, as the visualization lowers the barrier to entry into a mathematically challenging field. This is especially needed due to its interdisciplinarity and the various different academical backgrounds of learners~\cite{PhysRevPhysEducRes.18.010150}. Dimensional notations have several advantages over standard notations on a conceptual level. The dimensionality could make the effect of measurements and unitary operations in two- and three-qubit systems more intuitive due to a geometric depiction of single qubits as parts of these systems. As we showed in this paper, by introducing dimensionality, one can visualize separability and entanglement due to the ratio characterization (Corollary~\ref{cor:ratios}). This enables a new perspective on how entanglement is utilized in various quantum algorithms and processes. % and has applications beyond education. % as we discuss in \ref{sec:conclusions}.

It is important to consider the conceptual limitations of such explicit dimensional notations. First of all, larger than six- to seven-qubit systems will be difficult to visualize due to the exponential scaling of the number of basis states. An important limitation of these notations is that they can not completely replace mathematics for two reasons. Firstly, exact numerical amplitudes and phases are not visible, which, for example, means that many separable states can only \textit{approximately} be identified as such. Secondly, if variables are not displayed, one is restricted to specific examples. However, specific examples are often enough and even needed to understand the general case by abstraction.

Following the discussed limitations, we are working on developing an interactive web tool which makes it possible for everyone to visualize quantum operations in DCN. The repositories for this project can be found here: \url{https://github.com/QuanTUK/} and the website can be accessed via \url{https://dcn.physik.rptu.de/}.

We showed that $2-2^{n-1}$-separability is easily visualized in dimensional notations which is enough for all applications in up to three qubit systems and many applications beyond. In addition, Corollary~\ref{cor:ratios} can be used for e.g. $4-4$-separability in four-qubit systems as is discussed and visualized in Appendix~\ref{sec:4-4-sep}. Furthermore, we show in Appendix~\ref{sec:modular} that we can visualize quantum algorithms of up to at least five qubits as is shown there for a four-qubit error detection and a five-qubit error correction algorithm. For this, we ``modularize" DCN, arranging qubit systems in a variety of different ways to lay focus on specific entanglement properties and/or the geometry of unitary operations. By doing so, we aim to enhance understanding of complex multi-qubit algorithms and processes as for example the five-qubit error correction algorithm that utilizes a combination of quantum entanglement and classical correlation.

Another possible extension is the visualization of qudit-systems (qudits can be in $d$ possible states instead of only two). Gates and known algorithms in qudit systems are described in~\cite{10.3389/fphy.2020.589504}. Although qudits are not in the general focus of quantum computing at the moment, it is possible that they could be relevant sooner rather than later as there are some recent advancements 
\cite{Chi2022, roy2022realization,PRXQuantum.4.030327}. In this context, Corollary~\ref{cor:ratios} can be applied analogously to reveal entanglement properties of such systems.

We conclude that dimensional notations can find educational use in introductory quantum computing and quantum technology courses as well as in contexts beyond education to visualize the entanglement properties of multi-qubit systems complementary to the mathematical formalism. They provide a new perspective on entanglement in multi-qubit systems and how it is utilized in various quantum algorithms and by doing so, paired with the flexibility shown in Appendix~\ref{sec:modular}, could enhance understanding of these algorithms in general.

\section*{Acknowledgements}We thank Stefan Heusler from the WWU Münster for valuable general discussions and specific input regarding basis dependency of qubit models.

We also thank Steffen Glaser from the Technical University of Munich and Bettina Just from the Technische Hochschule Mittelhessen for helpful discussions.

M. K-E., P. L. and A. W. acknowledge support by the Quantum Initiative Rhineland-Palatinate (QUIP) and by the Research Initiative Quantum Computing for Artificial Intelligence (QC-AI).

J.B., E.R., A.A., M. K-E., P.L. and A.W. acknowledge support by the project QuanTUK at the RPTU in Kaiserslautern, supported by the Federal Ministry of Education and Research (FKZ13N15995).

N.L., L.K. and P.L. acknowledge support by the project KI4TUK at the RPTU in Kaiserslautern, supported by the Federal Ministry of Education and Research (BMBF) under grant number 16DHBKI058.

A.D., S.K. and J.K acknowledge support by the project Quantum Lifelong Learning (QL3) at the LMU Munich, supported by the Federal Ministry of Education and Research (BMBF) under grant number 13N16024, and the project DigiQ (EU), supported by the European Union's Digital Europe programme under grant number 101084035.

\appendix
\section{Single-qubit operations in circle notation}\label{sec:CNsinglequbit}

To understand single-qubit operations in multi-qubit systems in dimensional notations, it is enough to understand these operations in single-qubit systems which is one of the main advantages of dimensional notations in comparison to standard notations such as the circle notation. Fig.~\ref{fig:single_qubit_operations} shows some important single-qubit operations in single-qubit systems in circle notation.

\begin{figure}[htb]
    \centering
    \includegraphics[width=0.5\textwidth]{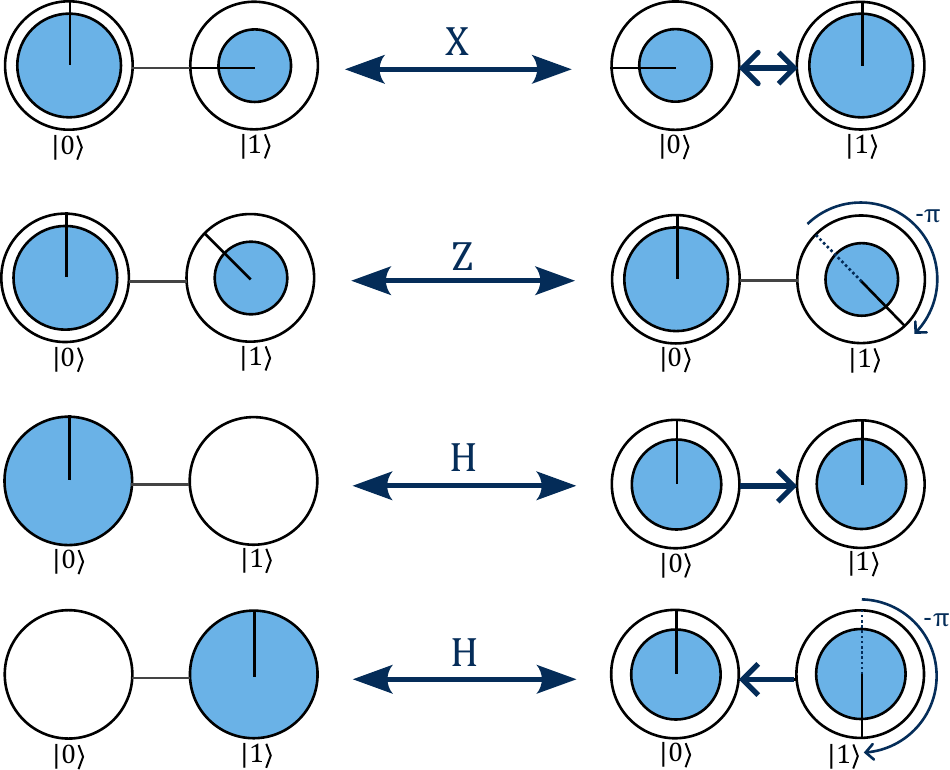}\caption{Single qubit operations in circle notation~\cite{johnston_harrigan_gimeno-segovia_2019}. The $X$-gate flips the coefficients of two states. The $Z$-gate adds a $+\pi$ phase to the $\ket{1}$-state, flipping the sign of the coefficient. The Hadamard-gate splits a state into two, flipping the phase if starting at $\ket{1}$. All these gates are self-adjoint, i.e. their own inverse.}
    \label{fig:single_qubit_operations}
\end{figure}

\section{Separability criteria for pure states}\label{sec:sep}

The following Lemma gives a necessary criterion for PQ-separability of any pure state ~\cite{doi:10.1142/S0129054103002035, dyn_ent}. We provide an alternative proof.

\textbf{Lemma 1:} Let $\ket{\psi_N}=\sum_{i=0}^{N-1}c_i\ket{i}$ $\in\mathcal{H}_N$ be a $P-Q$-separable pure state and let $r_0\in[0,Q-1]$, $k_0\in[0,P-1]$. Then, for all $r\in[0,Q-1]$ and $k\in[0,P-1]$: $c_{k_0Q+r_0}c_{kQ+r}=c_{k_0Q+r}c_{kQ+r_0}$.

\begin{proof}
For examination of $PQ$-separability, we write $\ket{\psi}$ as a matrix with $P$ rows and $Q$ columns:
    \begin{align}
    \ket{\psi} =&
    \begin{bmatrix}
        c_0 & c_1 & \ldots & c_{Q-1} \\
        c_{Q} & c_{Q+1}  & \ldots & c_{2Q-1} \\
        \vdots & \vdots & \ddots & \vdots \\
        c_{(P-1)\cdot Q} & c_{(P-1)\cdot Q+1} & \ldots & c_{P\cdot Q-1}
    \end{bmatrix}\\
    =:&
    \begin{bmatrix}
        \gamma_{0\,0} & \ldots & \gamma_{0\,Q-1} \\
        \vdots & \ddots & \vdots \\
        \gamma_{P-1\,0} & \ldots & \gamma_{P-1\,Q-1}
    \end{bmatrix}
\end{align}

with $\gamma_{k\,r}=c_{kQ+r}$ (don't confuse with density matrix!). 

$\ket{\psi}$ is separable into $\ket{\psi}=\left(\sum_{i=0}^{P-1}\alpha_i\ket{i}\right)\otimes\left(\sum_{j=0}^{Q-1}\beta_j\ket{j}\right)$ if and only if we can write it as

\begin{equation}\label{eq:psi_sep}
    \ket{\psi}=
    \begin{bmatrix}
        \alpha_0\beta_0 & \alpha_0\beta_1 & \ldots & \alpha_0\beta_{Q-1} \\
        \alpha_1\beta_0 & \alpha_1\beta_1  & \ldots & \alpha_1\beta_{Q-1} \\
        \vdots & \vdots & \ddots & \vdots \\
        \alpha_{P-1}\beta_0 & \alpha_{P-1}\beta_1 & \ldots & \alpha_{P-1}\beta_{Q-1}
    \end{bmatrix}
\end{equation}

We notice that we can draw matching diagonals in the matrix, i.e. $\gamma_{k_0\,r_0}\gamma_{k\,r}=\gamma_{k_0\,r}\gamma_{k\,r_0}$ for all $k\in[0,P-1]$ and $r\in[0,Q-1]$, because $\alpha_{k_0}\beta_{r_0}\alpha_k\beta_r=\alpha_{k_0}\beta_r\alpha_k\beta_{r_0}$. Now, with $\alpha_k\beta_r=c_{rQ+k}$, the statement follows.
\end{proof}

This means we have simple criteria for entanglement: For example, A pure state $\ket{\psi}\in\mathcal{H}_N$ is entangled if, for any $P$ and $Q$ with $N=PQ$, we find $k\in[1,P-1]$ and $r\in[1,Q-1]$ such that $c_0c_{kQ+r}\neq c_rc_{kQ}$.

The following Theorem adds conditions such that Lemma~\ref{lem} is also sufficient~\cite{doi:10.1142/S0129054103002035}. We provide an alternative proof.

\textbf{Theorem 2:}
    Let $\ket{\psi_N}=\sum_{i=0}^{N-1}c_i\ket{i}$ be a pure state where for some $i_0=k_0Q+r_0\in[0,N-1]$, $\alpha_{i_0}\neq0$ and $\forall i < i_0, \alpha_i=0$. Then, $\ket{\psi}$ is $N=PQ$-separable if and only if for all $k\in[k_0+1,P-1]$ and $r\in[0,Q-1]$: $c_{k_0Q+r_0}c_{kQ+r}=c_{k_0Q+r}c_{kQ+r_0}$.

\begin{proof}
    $"\Rightarrow"$: If $\ket{\psi}$ is separable, Lemma \ref{lem} says that the above condition has to hold. \\
    $"\Leftarrow"$: $\ket{\psi}$ looks like this:

\begin{equation}
    \ket{\psi} =
    \begin{bmatrix}
        0 & \ldots & 0  & 0 & \ldots & 0 \\
        \vdots & \ddots & \vdots  & \vdots & \ddots & \vdots \\
        0 & \ldots & 0  & 0 & \ldots & 0 \\
        0 & \ldots & 0 & \gamma_{k_0\,r_0} & \ldots & \gamma_{k_0\,Q-1}\\
        \gamma_{k_0+1\,0} & \ldots & \gamma_{k_0+1\,r_0-1} & \gamma_{k_0+1\,r_0}& \ldots & \gamma_{k_0+1\,Q-1}\\
        \vdots & \ddots & \vdots & \vdots & \ddots & \vdots \\
        \gamma_{P-1\,0} & \ldots & \gamma_{P-1\,r_0-1}& \gamma_{P-1\,r_0} & \ldots & \gamma_{P-1\,Q-1}
    \end{bmatrix}
\end{equation}

    with $\gamma_{r_0\,k_0}\neq0$. Due to the condition $c_{k_0Q+r_0}c_{kQ+r}=c_{k_0Q+r}c_{kQ+r_0}\Leftrightarrow \gamma_{k_0 \, r_0}\gamma_{k\,r}=\gamma_{k_0 \, r}\gamma_{k\,r_0}$, also for $r<r_0$, $\ket{\psi}$ can be written as

\begin{equation}\label{eq:psi_0}
    \ket{\psi} =
    \begin{bmatrix}
        0 & \ldots & 0  & 0 & \ldots & 0 \\
        \vdots & \ddots & \vdots  & \vdots & \ddots & \vdots \\
        0 & \ldots & 0  & 0 & \ldots & 0 \\
        0 & \ldots & 0 & \gamma_{k_0\,r_0} & \ldots & \gamma_{k_0\,Q-1}\\
        0 & \ldots & 0 & \gamma_{k_0+1\,r_0}& \ldots & \gamma_{k_0+1\,Q-1}\\
        \vdots & \ddots & \vdots & \vdots & \ddots & \vdots \\
        0 & \ldots & 0 & \gamma_{P-1\,r_0} & \ldots & \gamma_{P-1\,Q-1}
    \end{bmatrix}.
\end{equation}
    Also, $\gamma_{k_0 \, r_0}\gamma_{k\,r}=\gamma_{k_0 \, r}\gamma_{k\,r_0}$ is equivalent to i) $\gamma_{k\, r}=0$ and a) $\gamma_{k_0\,r}=0$ or b) $\gamma_{k\,r_0}=0$ (or both) or ii) $\gamma_{k\,r}=m_r\gamma_{k\,r_0}$ with $m_r=\frac{\gamma_{k_0\,r}}{\gamma_{k_0\,r_0}}$. 
    
    Let's discuss the case i), $\gamma_{k \, r}=0$ for specific $k\in[k_0+1,P-1]$ and $r\in[r_0,Q-1]$. In the case of a), $\gamma_{k_0\,r}=0$, due to $\gamma_{k_0 \, r_0}\gamma_{k\,r}=\gamma_{k_0 \, r}\gamma_{k\,r_0}$ for all $k\in[k_0+1,P-1]$, $\gamma_{k\,r}=0$ has to hold for all $k\in[k_0+1,P-1]$, meaning the whole column has to be 0. In this case, with $m_r=0$, $\gamma_{k\,r}=m_r \gamma_{k\,r_0}$ also holds. In the other case b), $\gamma_{k\, r_0}=0$, $\gamma_{k\,r}=0$ has to hold for all $k\in[k_0+1,P-1]$, meaning the whole row has to be 0. In this case, $\gamma_{k\,r}=m_r\gamma_{k r_0}$ still holds for all $k\in [k_0+1,P-1]$ regardless of how $m_r$ is chosen, because in this case $\gamma_{k\, r_0},\gamma_{k\, r}=0$. In summary, in all cases, $\ket{\psi}=$

\begin{equation}\label{eq:ratios}
    \begin{bmatrix}
        0 & \ldots & 0 &0 & 0 & \ldots & 0 \\
        \vdots & \ddots & \vdots& \vdots & \vdots & \ddots & \vdots \\
        0 & \ldots & 0 &0 & 0 & \ldots & 0 \\
        0 & \ldots & 0 & \gamma_{k_0\,r_0} & m_{r_0+1}\gamma_{k_0\,r_0} & \ldots & m_{Q-1}\gamma_{k_0\,r_0}\\
        0 & \ldots & 0 & \gamma_{k_0+1\,r_0}& m_{r_0+1}\gamma_{k_0+1\,r_0} &\ldots & m_{Q-1}\gamma_{k_0+1\,r_0}\\
        \vdots & \ddots & \vdots & \vdots & \ddots & \ddots & \vdots \\
        0 & \ldots & 0 & \gamma_{P-1\,r_0} & m_{r_0+1}\gamma_{P-1\,r_0} & \ldots & m_{Q-1}\gamma_{P-1\,r_0}
    \end{bmatrix}.
\end{equation}

Writing separable states in the form of Eq.~\ref{eq:ratios} is the central idea of the ratio characterization that, as we show in this paper, can be visualized in dimensional notations. % In fact, $\ket{\psi}$ is separable if and only if we find ratios $m_r$ such that it can be written in this form (Corollary~\ref{cor:ratios}).

We now choose $\alpha_k=0$ for all $k<k_0$ and $\beta_r=0$ for all $r<r_0$, $\alpha_{k_0}=\gamma_{k_0\,r_0}$ and $\beta_{r_0}=1$. Additionally, we choose $\beta_{r}=m_r$ for all $r\in[r_0+1,Q-1]$ and $\alpha_{k}=\gamma_{k\,r_0}$ for all $k\in[k_0+1,P-1]$. Then, $\ket{\psi}$ can be written in the same form as Eq. \ref{eq:psi_sep}.
\end{proof}

Note that this means that Lemma \ref{lem} is also a sufficient criterion if $c_0\neq 0$. The following is another alternative formulation of Theorem \ref{them:main}, requiring the terms where $r<r_0$ all to be 0:

\begin{corollary}\label{cor:extra0}
     Let $\ket{\psi_N}=\sum_{i=0}^{N-1}c_i\ket{i}$ be a pure state where for some $i_0=k_0Q+r_0\in[0,N-1]$, $\alpha_{i_0}\neq0$ and $\forall i < i_0, \alpha_i=0$. Then, $\ket{\psi}$ is $N=PQ$-separable if and only if for all $k\in[k_0+1,P-1]$ and $r\in[r_0+1,Q-1]$: $c_{k_0Q+r_0}c_{kQ+r}=c_{k_0Q+r}c_{kQ+r_0}$ and for all $r<r_0$, $c_{kQ+r}=0$.
\end{corollary}

\begin{proof}
    Analogous to proof of Theorem \ref{them:main} where one starts with $\ket{\psi}$ in the form of Eq.~\ref{eq:psi_0}.
\end{proof}

Alternatively, Lemma \ref{lem} can be reformulated exchanging $c_0$ with some non-zero $c_{i_0}$ to be sufficient:

\begin{corollary}
    Let $\ket{\psi_N}=\sum_{i=0}^{N-1}c_i\ket{i}$ be a pure state and let $i_1=k_1Q+r_1\in[0,N-1]$ with $c_{i_1}\neq0$. Then, $\ket{\psi}$ is $N=PQ$-separable if and only if for all $k\in[0,P-1]$ and $r\in[0,Q-1]$: $c_{i_1}c_{kQ+r}=c_{k_1Q+r}c_{kQ+r_1}$.
\end{corollary}

\begin{proof}
    "$\Rightarrow$" is given by Lemma~\ref{lem}. 

    "$\Leftarrow$": Let $c_{i_0}$ be the first non-zero coefficient of $\ket{\psi}$ ($i_0\leq i_1$). We can then write $\ket{\psi}$ as in Eq.~\ref{eq:psi_0}. Analogous to the proof of Theorem~\ref{them:main}, we can then write $\ket{\psi}$ in the form of Eq.~\ref{eq:ratios}.
\end{proof}

The following ratio characterization of separability is used throughout this paper to visualize entanglement. It is part of the proof of Theorem~\ref{them:main}.

The following can be used as a special case of Corollary~\ref{cor:ratios} to separate single qubits from systems (after renumbering).

\begin{corollary}[2-2$^{n-1}$-Separability]
    \label{cor:par_sep}
    Let $\alpha,\beta,c_i\in\mathbb{C}$. An $n$-qubit state $\ket{\psi}=\sum_{i=0}^{2^n-1}c_i\ket{i}$ is 2-2$^{n-1}$ separable into $\ket{\psi}=(\alpha\ket{0}+\beta\ket{1})\otimes\sum_{i=0}^{2^{n-1}-1} c'_i\ket{i}$ if and only if for all $i\in\{0,\ldots,2^{n-1}-1\}$ either $c_{i}=0$ or there exists ratios $m_i\in\mathbb{C}$ such that $c_{2^{n-1}+i}=m_ic_{i}$.
\end{corollary}

\begin{proof}
    Analogous to proof of Theorem \ref{them:main}.
\end{proof}

In dimensional notations, Corollary~\ref{cor:par_sep} visually translates to looking at this exact ratio condition along the axes of one qubit to determine whether that qubit is separable from the system. Other use cases of Corollary~\ref{cor:ratios} are higher order cases of separability in larger than three qubit systems. We provide the example of visualizing 4-4 separability in Sec.\ref{sec:4-4-sep}.

\section{Multi-Qubit Gates and Algorithms in two-Qubit Systems}\label{sec:DCNmulti}

Phase kickback is an inherently quantum concept and an essential part of quantum computing. The main idea is that by local basis transformation, operations with a control and a target qubit are inverted such that the roles of control and target qubit are swapped. This happens because the control qubit inherits the phase of the target qubit while the target qubit is unchanged. This has applications in, e.g., so-called oracle functions that are part of many quantum algorithms -- the controlled gates are applied to a set of auxiliary qubits in the Hadamard basis, such that the logical qubits are changed~\cite{Lee_2016}. Fig.~\ref{fig:phasekickback} shows the most basic example of a phase kickback and Fig.~\ref{fig:deutsch} shows a use case of this: the Deutsch algorithm.

\begin{figure*}[htb]
    \centering
    \includegraphics[width=0.8\textwidth]{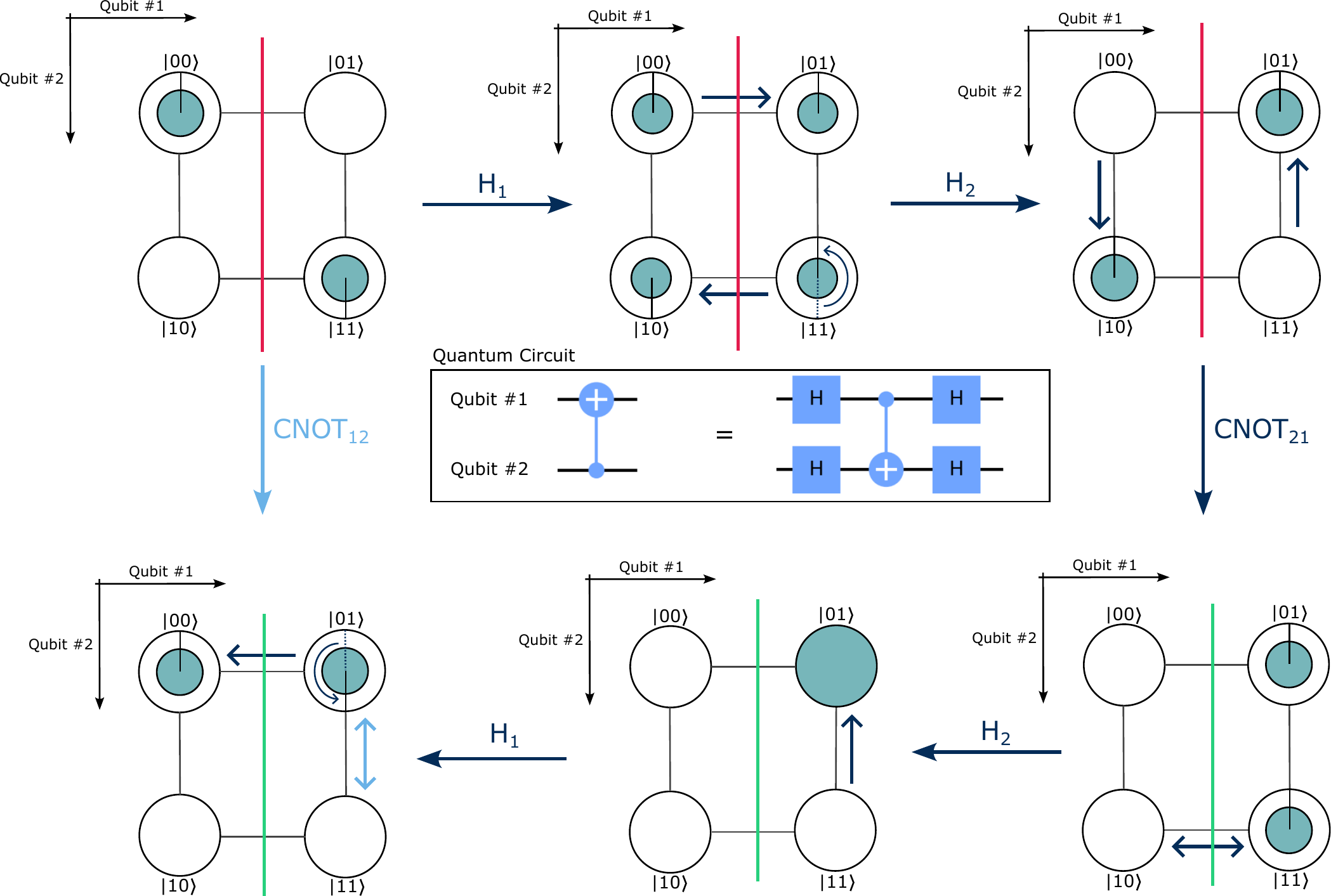}
    \caption{Basic phase kickback, i.e. the relation $\text{CNOT}_{12}=(\text{H}_2\otimes \text{H}_1)\text{CNOT}_{21}(\text{H}_2\otimes \text{H}_1)$, shown with the initial state $\ket{\psi}=1/\sqrt{2}(\ket{00}-\ket{11})$. The change of basis into the Hadamard basis by applying Hadamard-gates on all qubits makes the CNOT$_{21}$-gate work like a CNOT$_{12}$-gate. After application of the CNOT$_{12}$-gate, the state changes from entangled to separable as indicated by the green and red axes.}
    \label{fig:phasekickback}
\end{figure*}

\begin{figure*}[htb]
    \centering
    \includegraphics[width=0.9\textwidth]{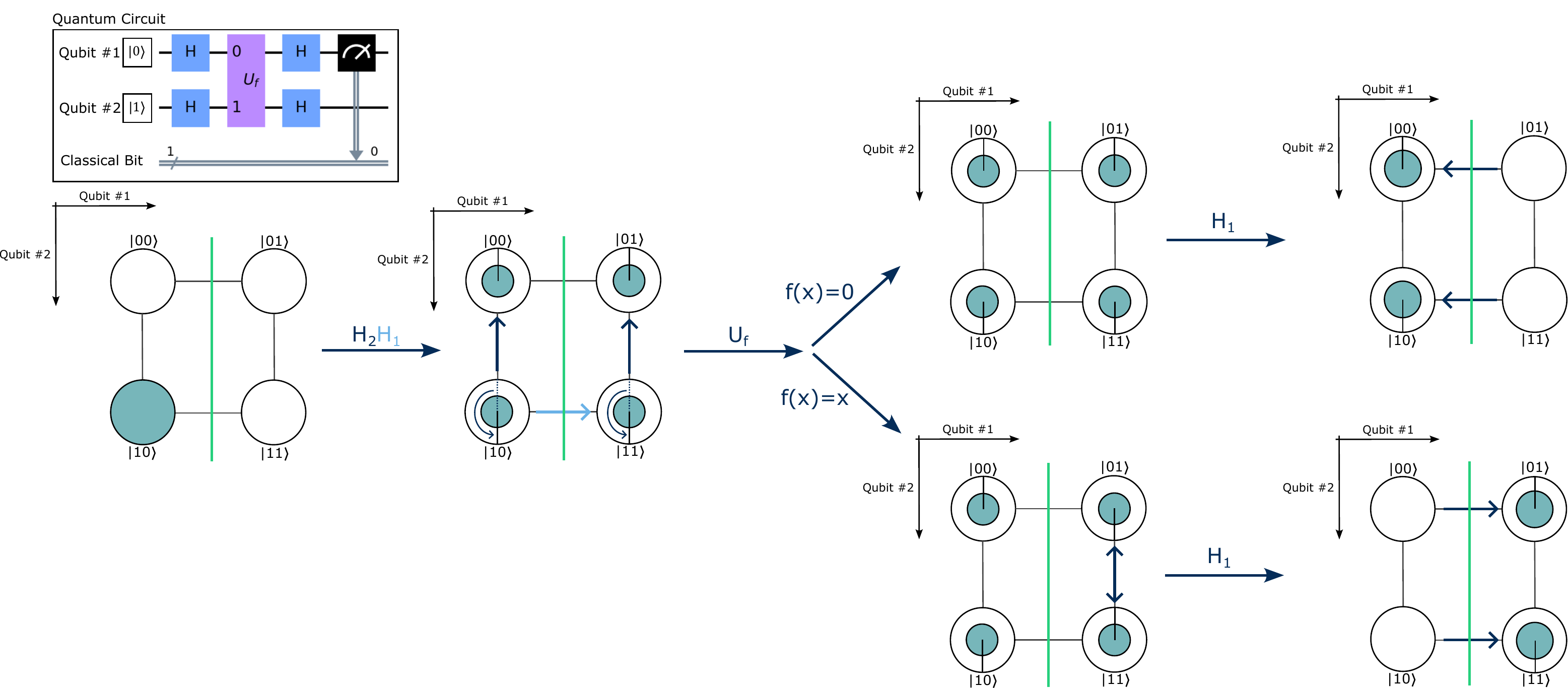}
    \caption{The Deutsch algorithm to determine whether a function $f:\{0,1\}\to \{0,1\}$ is constant ($f=0$ or $f=1$) or balanced ($f(x)=x$ or $f(x)=x\oplus1$ where $1\oplus1=0$). The Qubits are initialized to the state $\ket{10}$. After application of Hadamard-Gates on all qubits, the system is in equal superposition with a phase shift on qubit~\#2. Then the oracle $U_f$ defined by $U_f:\ket{x}\ket{y}\to\ket{x}\ket{f(x)\oplus y}$ is applied. The two cases where $f$ is constant and the two cases where $f$ is balanced only differ by a global phase, respectively. Therefore, only the cases $f=0$ and $f(x)=x$ are shown. After application of a Hadamard-Gate on qubit~\#1, one can see that the operation $U_f$ actually acted on qubit~\#1 due to phase kickback. When measuring qubit~\#1, the result will be 0 when $f$ was balanced and 1 when $f$ was constant. At all points, the system remains separable as shown by the green symmetry axes, showing that the Deutsch algorithm does not utilize entanglement and can, in fact, be implemented classically as shown in \cite{Kish2023} experimentally.}
    \label{fig:deutsch}
\end{figure*}

Fig. \ref{fig:qt_measurement} shows the last step of the quantum teleportation algorithm. The system is fully entangled before the measurement. 

\begin{figure*}[htb]
    \centering
    \includegraphics[width=\textwidth]{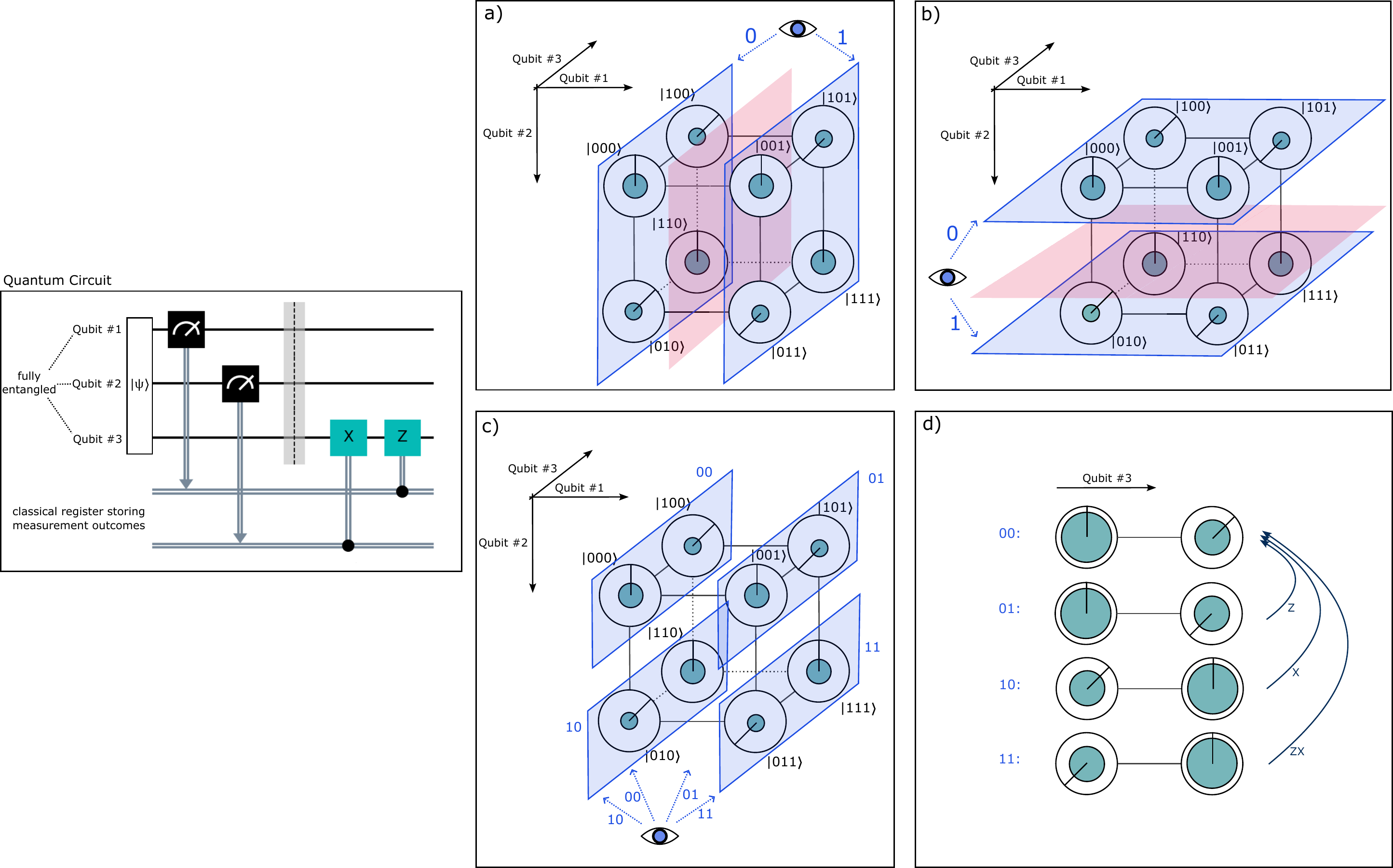}
    \caption{The last steps of the quantum teleportation process: Alice measures and sends the information to Bob who then applies single-qubit gates according to the measurement result. The corresponding quantum circuit is displayed on the left. The system starts in the fully entangled state $2\ket{\psi}=(\sqrt{2/3}\ket{0}+1/\sqrt{3}e^{-i\pi/4}\ket{1})\ket{00}+(\sqrt{2/3}\ket{0}+1/\sqrt{3}e^{3i\pi/4}\ket{1})\ket{01}+(1/\sqrt{3}e^{-i\pi/4}\ket{0}+\sqrt{2/3}\ket{1})\ket{10}+(1/\sqrt{3}e^{3i\pi/4}\ket{0}+\sqrt{2/3}\ket{1})\ket{11}$ also shown in Fig.~\ref{fig:qt_sum}. Alice measures qubit~\#1 and \#2. a) The measurement of qubit~\#1 and a red plane showing that there is not symmetry along the axis of qubit~\#1, i.e. it is not separable from the system; b) The measurement of qubit~\#2 and, again, a red plane indicating the inseparability from the system. c) The combined measurement of qubit~\#1 and \#2. Because the sum of the areas of the inner circles is the same for all of the four possibilities, the chance of measuring any of the four values is 25\%. d) The four possible states of qubit~\#3 depending on the measurement outcome. Bob has to apply an $X$ and/or a $Z$-gate such that qubit~\#3 is in the previous state of qubit~\#1.}
    \label{fig:qt_measurement}
\end{figure*}

\section{Visualizing 4-4-Separability in four-Qubit Systems}\label{sec:4-4-sep}

Let $\ket{\psi}=\sum_{i=0}^{15}c_i\ket{i}$ be a pure four-qubit state. $2-8$-separability is easily examined in this system using Corollary~\ref{cor:par_sep} which states that we can apply two planes to the two hypercubes or a plane between the hypercubes and look for a ratio along these planes in order to examine separability characteristics. To examine 4-4-separability, we write $\ket{\psi}$ as 

\begin{equation}
    \ket{\psi}=\begin{bmatrix}
        c_{0000} & c_{0001} & c_{0010} & c_{0011} \\
        c_{0100} & c_{0101} & c_{0110} & c_{0111} \\
        c_{1000} & c_{1001} & c_{1010} & c_{1011} \\
        c_{1100} & c_{1101} & c_{1110} & c_{1111} \\
    \end{bmatrix}.
\end{equation}

Let $c_{0000}\neq0$. Then, Corollary~\ref{cor:ratios} states that $\ket{\psi}$ is 4-4-separable if and only if it can be written as 

\begin{equation}
    \ket{\psi}=\begin{bmatrix}
        c_{0000} & m_1c_{0000} & m_2c_{0000} & m_3c_{0000} \\
        c_{0100} & m_1c_{0100} & m_2c_{0100} & m_3c_{0100} \\
        c_{1000} & m_1c_{1000} & m_2c_{1000} & m_3c_{1000} \\
        c_{1100} & m_1c_{1100} & m_2c_{1100} & m_3c_{1100} \\
    \end{bmatrix}.
\end{equation}

This is visualized in figure \ref{fig:4-4-sep}. If $c_{0000}=0$, then $c_{xy00}=0$ has to hold for the state to be separable (similarly in the case of $c_{0001}=0$: $c_{xy01}=0$, see the condition in Corollary~\ref{cor:extra0}). Then, we only need to account for two ratios, e.g. going down and going left in both cubes.

\begin{figure*}[htb]
    \centering
    \includegraphics[width=0.7\textwidth]{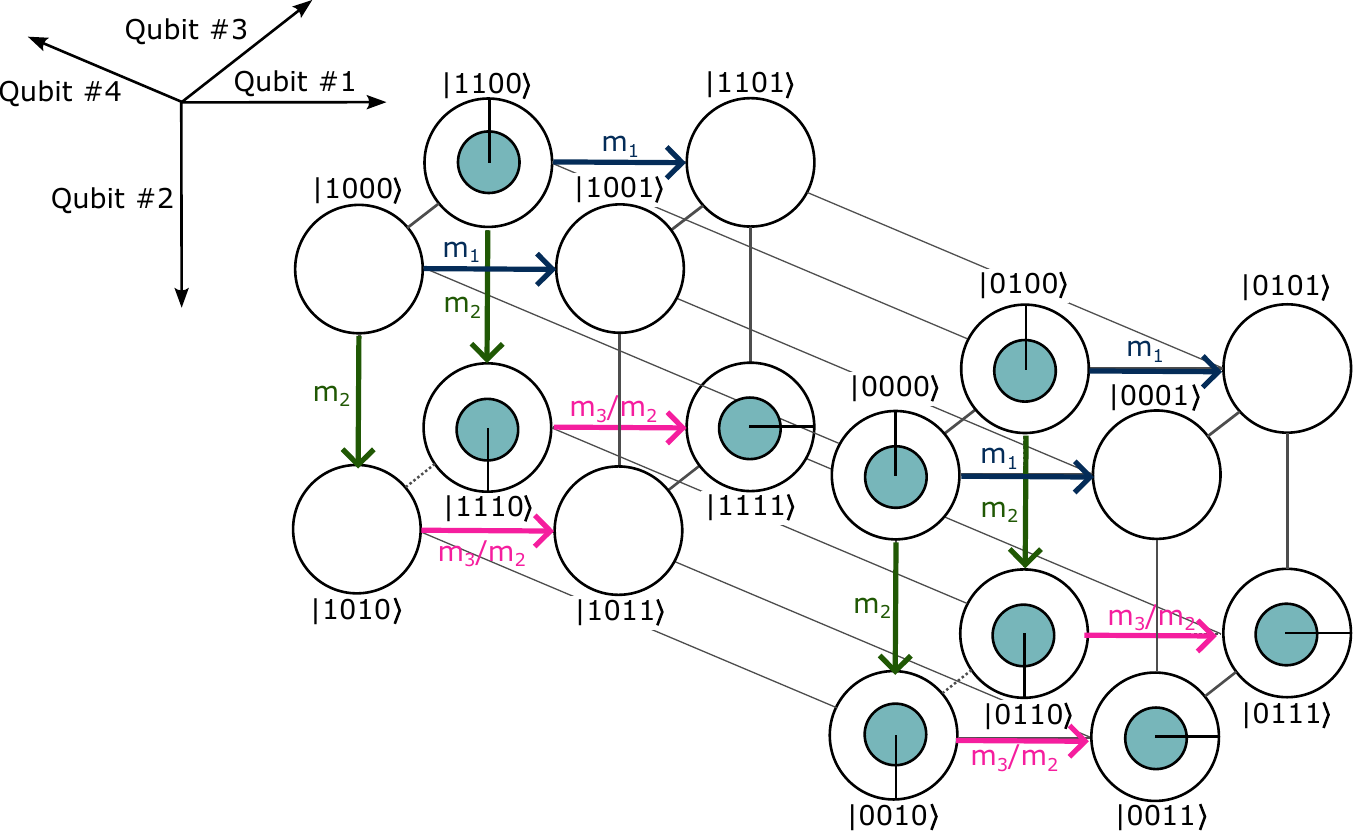}
    \caption{Examination the 4-4-separable state $\ket{\psi}=1/3(\ket{0000}-\ket{0010}+e^{-i\pi/2}\ket{0011}+\ket{0100}-\ket{0110}+e^{-i\pi/2}\ket{0111}+\ket{1100}-\ket{1110}+e^{-i\pi/2}\ket{1111})=1/3(\ket{00}+\ket{01}+\ket{11})\otimes(\ket{00}-\ket{10}+e^{i\pi/2}\ket{11})$ in DCN with Corollary~\ref{cor:ratios}. We find 4-4-separability with $m_1=0$, $m_2=-1$ and $m_3/m_2=e^{i\pi/2}$. Note that we can use $m_2\neq0$ to not need to look at the diagonal of $m_3$. Also note that there is no symmetry plane (as they would have to be applied to both cubes or between the cubes for separability of qubit~\#4), i.e. no single qubit can be separate from the system (Corollary~\ref{cor:par_sep}).}
    \label{fig:4-4-sep}
\end{figure*}

In principle, it is possible to go beyond 4-4 separability, to 4-8 separability in five-qubit systems (for this, do the same thing as in figure \ref{fig:4-4-sep} in two cubes) or even 8-8 separability in six-qubit systems (in this case, one would need to account for seven different ratios and each in seven different directions). These very complex entanglement properties are less and less easily spotted, but the process remains the same.

\section{Modular DCN in four- and five-qubit systems}\label{sec:modular}

In this section, we give examples on how to represent qubit ensembles of four and five qubits in various ways. There are multiple ways to represent four-qubit systems (systems with 16 basis states) in three dimensional space (and, on paper, then in two dimensions). One natural possibility is a projection of a four dimensional hypercube into three dimensions. This retains the geometric depiction of entanglement that is presented in this paper. For the ratio characterization of separability, eight pairs of coefficients have to be compared for each qubit in order to check for separability of that qubit from the system. 

In quantum settings, decoherence is a common factor to consider. Quantum Error correction can counteract the effects of decoherence. Classical error correction is often thought of in terms of hypercubes ~\cite{bellcore_aiello, doi:10.1080/00207160211287,  doi:10.1137/S0097539798332464}. In fact, similar ideas exist for quantum error correction as seen in hypercubes or hypercube-like lattices~\cite{Kubica_2015,Vasmer_2022}. Therefore, it makes sense to apply DCN to quantum error detection and correction. Here, we show the four-qubit error detection code demonstrated experimentally in~\cite{Corcoles2015} in Fig.~\ref{fig:four_qubit_error} in a hypercube. Note that for a code to also \textit{correct} the detected error, it needs five qubits to function~\cite{10.1093/nsr/nwab011}. This five qubit algorithm functions by entangling three qubits into a GHZ state, and then uses two anzilla qubits to correct the error. Interestingly, it can be seen in modular DCN that this algorithm does \textit{not} utilize quantum entanglement between the two subsystems qubit~\#1, \#2 and \#3 and the anzilla qubits qubit\#4 and \#5. Because the state of the system of the first three qubits depends on which error occurred which qubit~\#4 and \#5 depend on, it can be seen as a classical correlation between the two subsystems.

Another possibility is to represent the system using a mixture of circle notation and DCN that we call modular DCN.
We can have two or more qubits on every axis and assign only specific qubits to their own axis. We can then check, again via ratio characterization, separability from the system of the qubits that have their own axis. The five qubit error correction code that is shown in e.g.~\cite{kasirajan_2022} is visualized in Fig.~\ref{fig:3-qubit_ent} (simple three-qubit encoding process and three possible single-qubit flip errors), Fig.~\ref{fig:five_qubits_1} (transfer Syndrome and error correction in modular 2x2x8 DCN) and Fig.~\ref{fig:five_qubit_2} (the last step of error correction in a four-cube system). DCN is flexible as we can arrange qubit ensembles in modular DCN in a variety of different ways to lay focus on specific multi-partite entanglement properties and/or in a way such that the visualized unitary operations remain geometrically intuitive with the aim of enhancing understanding of complex multi-qubit algorithms.

\begin{figure*}
    \centering
    \includegraphics[width=\textwidth]{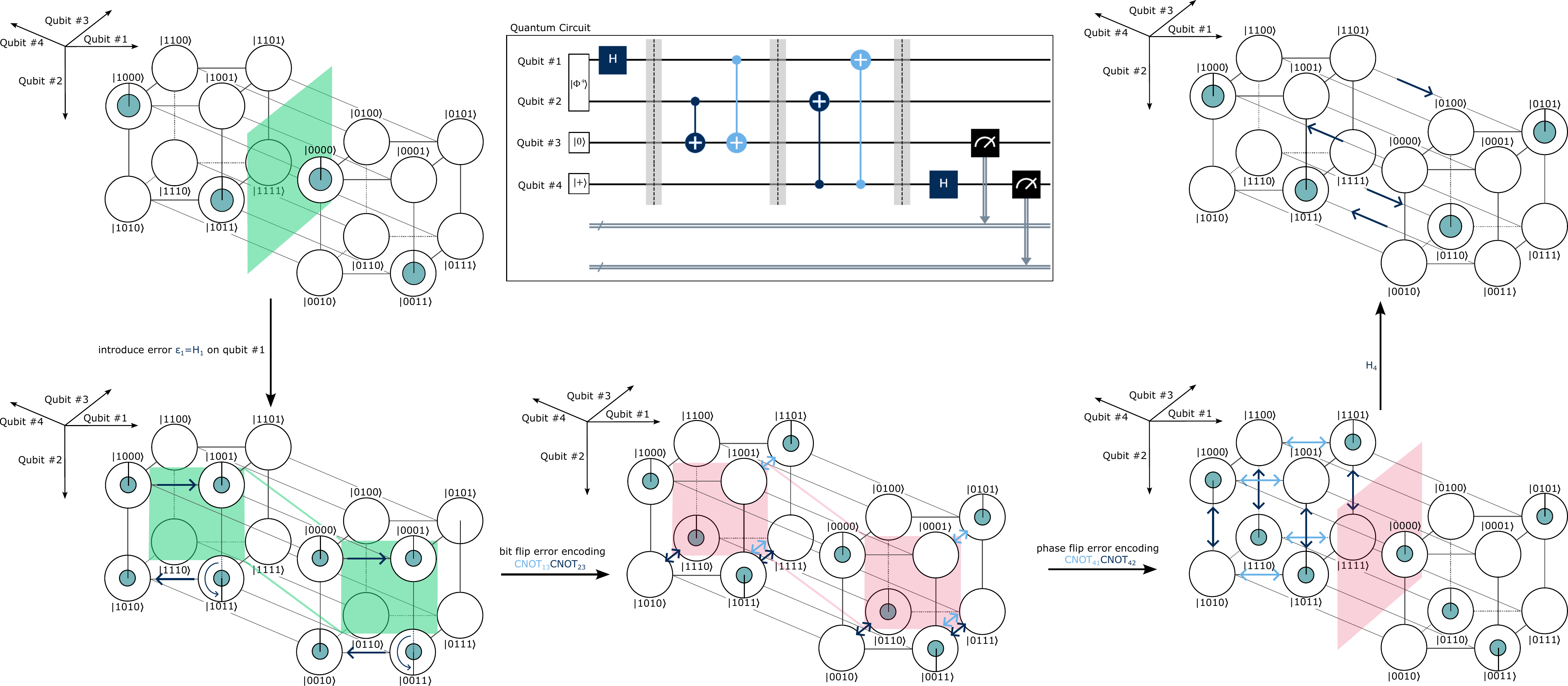}
    \caption{Four-qubit quantum error detection code as demonstrated experimentally in~\cite{Corcoles2015}, here in the case of a Hadamard error. The system is initialized to the state $\ket{\psi}=1/\sqrt{2}(\ket{0}+\ket{1})\otimes\ket{0}\otimes 1/\sqrt{2}(\ket{00}+\ket{11})$ where qubit~\#1 and \#2 are entangled and qubit \#4 is brought into the Hadamard basis $\ket{+}=1/\sqrt{2}(\ket{0}+\ket{1})$ in order to detect a phase flip. First, an error $\epsilon_1$ is applied, in this case a Hadamard error $H_1$ corresponding to half of a bit flip and half a phase flip on qubit~\#1. Then, the bit flip error is encoded onto qubit~\#3 via the CNOT$_{13}$CNOT$_{23}$ operation, entangling qubit~\#3 with qubit~\#1 and \#2. Afterwards, the operation CNOT$_{41}$CNOT$_{42}$ that can be seen as a 180$^\circ$ rotation of the cube corresponding to qubit \#4 being in the state 1 in the plane spanned by qubit~\#1 and \#2, fully (phase-)entangling the system. The Hadamard gate then turns this phase-entanglement into a magnitude entanglement in terms of qubit \#3 and \#4. In the end, qubit \#4 will be found in the state 1 if a phase flip has occurred while qubit~\#3 will be found in the state 1 when a bit flip has occurred. In this case of a Hadamard error, the error detection algorithm will always find that there was some error, as qubit~\#3 and \#4 are anti-correlated as can be seen in DCN.}
    \label{fig:four_qubit_error}
\end{figure*}

\begin{figure*}
    \centering
    \includegraphics[width=0.9\textwidth]{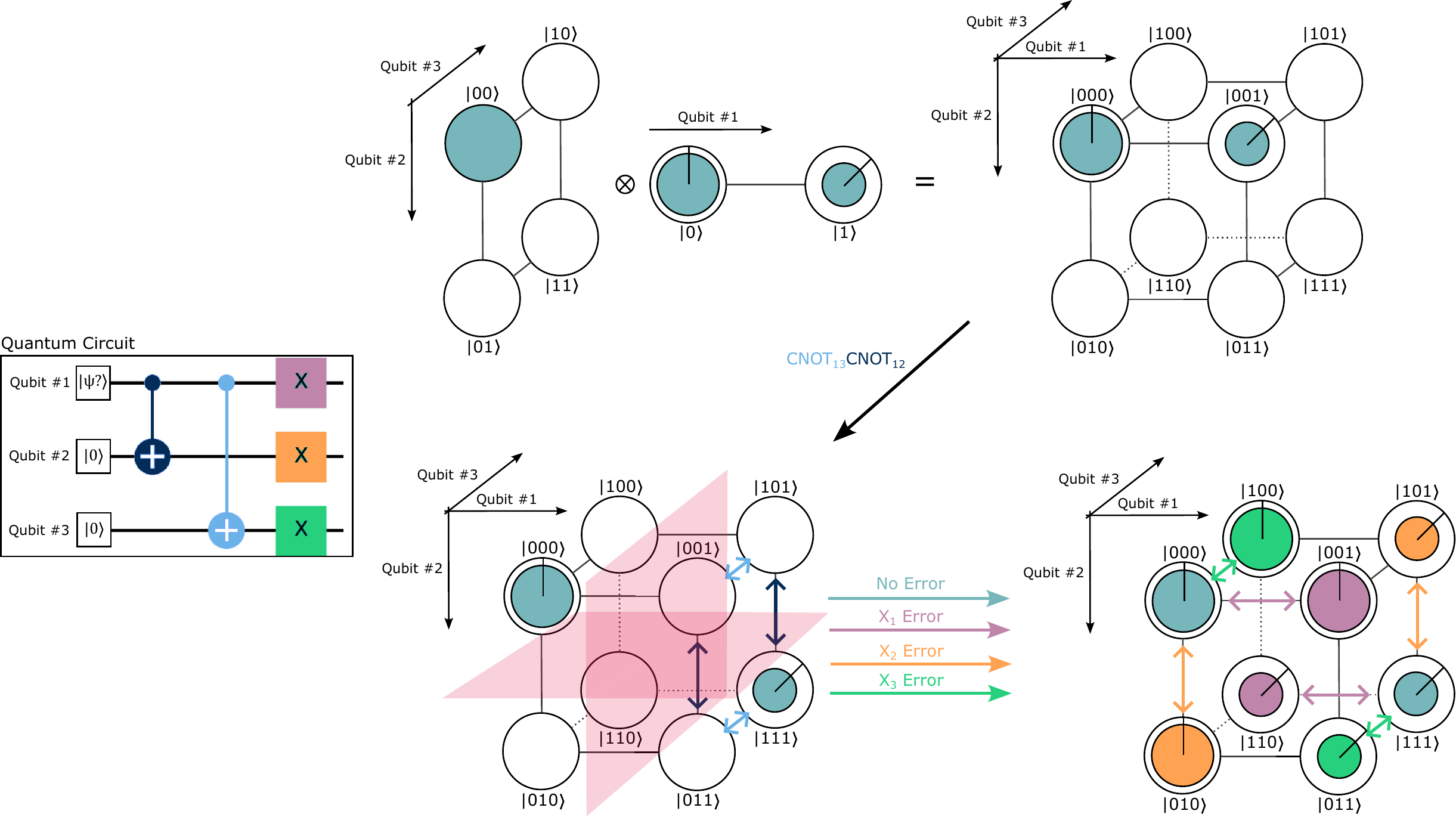}
    \caption{The initial step of error correcting the (arbitrary) state $\ket{\psi}_1=\sqrt{2}/\sqrt{3}\ket{0}+1/\sqrt{3}e^{-i\pi/4}\ket{1}$ using four additional qubits. First, qubit~\#1 is entangled with qubit~\#2 and \#3 in a GHZ-similar state $\ket{\psi}=\sqrt{2}/\sqrt{3}\ket{000}+1/\sqrt{3}e^{-i\pi/4}\ket{111}$ with two CNOT-gates. The system is fully entangled as can be seen by the lack of symmetry indicated by the red planes. Then, a bit flip error is applied. Here, three possible bit flip errors are shown (lilac = bit flip error on qubit~\#1, orange = bit flip error on qubit~\#2 and green = bit flip error on qubit~\#3) as well as the case of no bit flip errors in gray blue. We assume that only one bit flip error occurs at the same time. The bit flip errors do not change entanglement properties of the system.}
    \label{fig:3-qubit_ent}
\end{figure*}

\begin{figure*}
    \centering
\includegraphics[width=\textwidth]{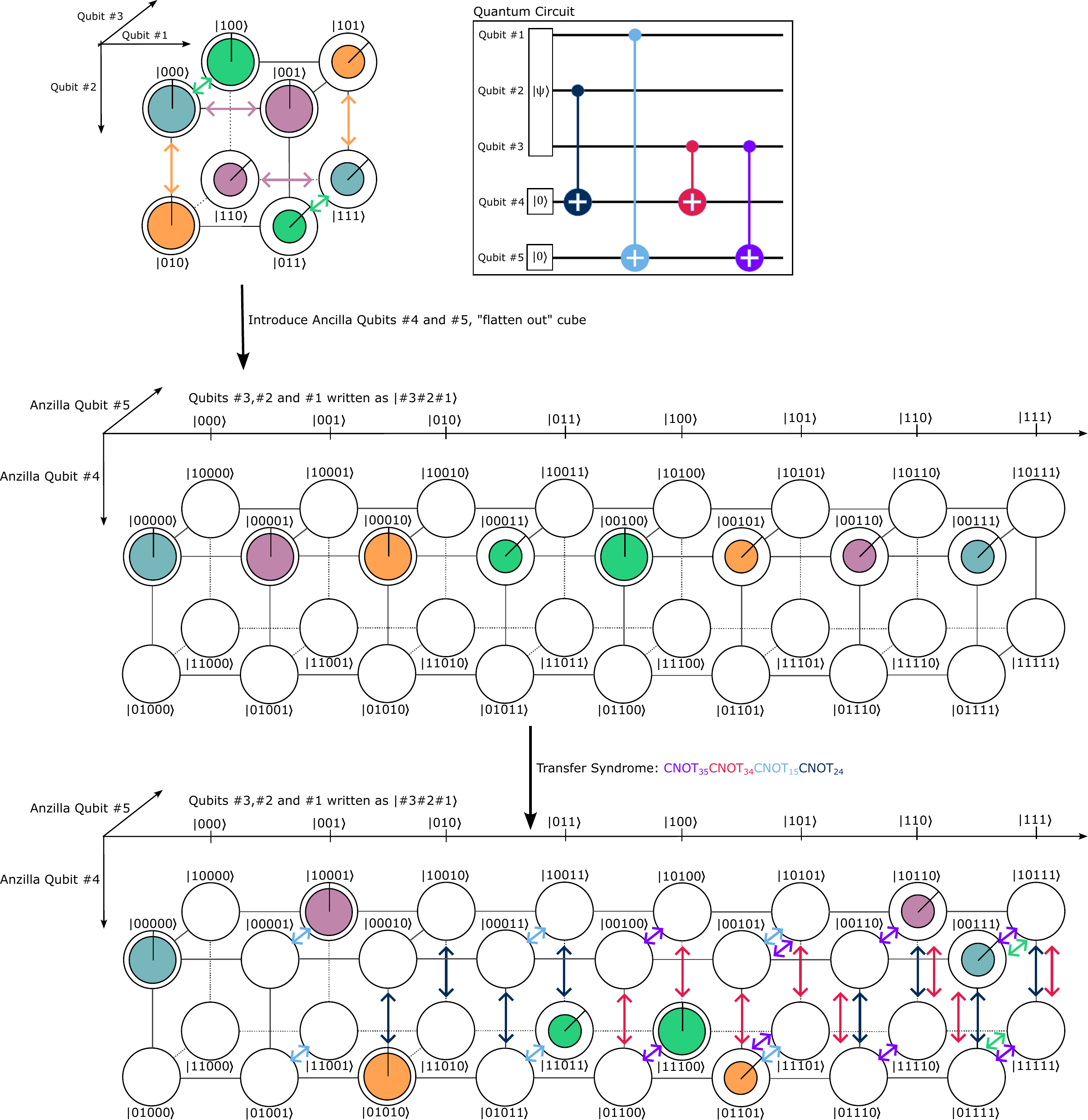}
    \caption{The "transfer syndrome" step of error correcting the state $\ket{\psi}_1=\sqrt{2}/\sqrt{3}\ket{0}+1/\sqrt{3}e^{-i\pi/4}\ket{1}$ using four additional qubits visualized in modular DCN. We start in the final state $\ket{\psi}$ of Fig.~\ref{fig:3-qubit_ent}, flatten out the cube to standard circle notation and introduce the anzilla qubits \#4 and \#5, arranging the system in modular DCN. The CNOT$_{24}$- and CNOT$_{34}$-gates encode an $X_2$-error onto anzilla qubit \#4 and the CNOT$_{35}$- and CNOT$_{15}$-gates encode an $X_1$-error onto anzilla qubit \#5 while an interesting and desirable byproduct of these operations is that an $X_3$-error is encoded on both anzilla qubits.}
    \label{fig:five_qubits_1}
\end{figure*}

\begin{figure*}
    \centering
\includegraphics[width=\textwidth]{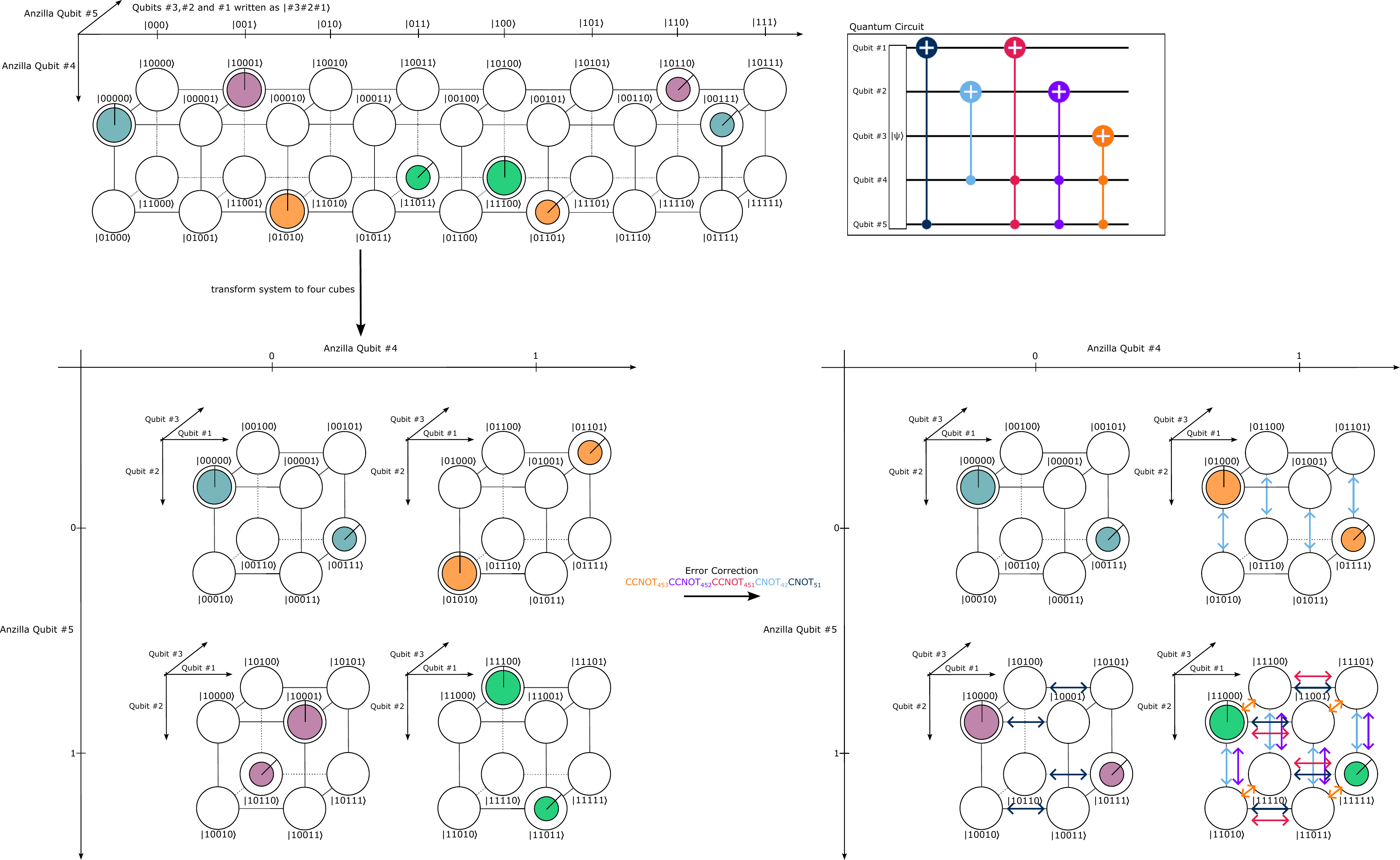}
    \caption{The last step of error correcting the state $\ket{\psi}_1=\sqrt{2}/\sqrt{3}\ket{0}+1/\sqrt{3}e^{-i\pi/4}\ket{1}$ visualized in modular DCN. We start by transforming the depiction of the last state $\ket{\psi}$ in Fig.~\ref{fig:five_qubits_1} to a four-cube system where the cubes are represented in space depending on anzilla qubit \#4 and \#5. Here, we can see that the three different kinds of bit flip errors correspond to three different configurations of anzilla qubits \#4 and \#5. Now, CNOT-gates are applied to correct these errors. The subsystem of qubit~\#1, \#2 and \#3 is classically correlated with the subsystem of qubit~\#4 and \#5, but interestingly, they are \textit{not} entangled even after application of the CNOT-gates (Fig.~\ref{fig:five_qubits_1})! Now, the CNOT$_{51}$ gate corrects the $X_1$-error, the CNOT$_{42}$-gate corrects the $X_2$-error and the CCNOT$_{453}$-gate corrects the $X_3$-error. Lastly, the CCNOT$_{452}$- and CCNOT$_{451}$-gates are needed to counteract the unwanted effects of the first two CNOT-gates in the case of an $X_3$-error. Qubit~\#4 and \#5 are still not entangled with the system. However, the entanglement between qubits~\#1, \#2 and \#3 is preserved. Now we can see that in all four cases, qubit~\#1 is in the desired state $\ket{\psi}_1$. Qubit~\#4 and \#5 can now be measured to see whether a bit flip error has occurred and which one.}
    \label{fig:five_qubit_2}
\end{figure*}
\FloatBarrier
\bibliography{main}
\end{document}